\documentclass[11pt]{article}
\usepackage{amsfonts}
\usepackage{amssymb}
\usepackage{fancyhdr}
\usepackage{comment}
\usepackage{prettyref}
\newrefformat{eq}{(\ref{#1})}
\newrefformat{sec}{Section~\ref{#1}}
\newrefformat{alg}{Algorithm~\ref{#1}}
\newrefformat{fig}{Fig.~\ref{#1}}
\newrefformat{tab}{Table~\ref{#1}}
\newrefformat{rmk}{Remark~\ref{#1}}
\newrefformat{clm}{Claim~\ref{#1}}
\newrefformat{def}{Definition~\ref{#1}}
\newrefformat{cor}{Corollary~\ref{#1}}
\newrefformat{lmm}{Lemma~\ref{#1}}
\newrefformat{prop}{Proposition~\ref{#1}}
\newrefformat{app}{Appendix~\ref{#1}}
\newrefformat{ass}{Assumption~\ref{#1}}
\newrefformat{hyp}{Hypothesis~\ref{#1}}
\newrefformat{thm}{Theorem~\ref{#1}}
\newrefformat{assump}{Assumption~\ref{#1}}
\newrefformat{conj}{Conjecture~\ref{#1}}

\usepackage[margin=1in]{geometry}  

\usepackage{xspace,prettyref}
\usepackage{amsmath}
\usepackage{amsthm}
\usepackage{dsfont}
\usepackage{graphicx}
\usepackage{enumerate}
\usepackage[hypcap]{caption}
\usepackage{color, soul}
\usepackage{float}
\usepackage[ruled]{algorithm2e}
\usepackage{subcaption}
\newtheorem{theorem}{Theorem}
\newtheorem{lemma}{Lemma}

\newtheorem{assumption}{Assumption}

\theoremstyle{definition}
\newtheorem{definition}{Definition}
\newtheorem{example}{Example}
\newtheorem{remark}{Remark}

\newcommand{\Red}{{\rm red}}
\newcommand{\Blue}{{\rm blue}}

\newcommand{\diverge}{\to\infty}

\newcommand{\reals}{{\mathbb{R}}}

\newcommand{\eexp}{{\rm e}}

\newcommand{\diff}{{\rm d}}

\newcommand{\Expect}{\mathbb{E}}
\newcommand{\expect}[1]{\mathbb{E}\left[ #1 \right]}

\newcommand{\prob}[1]{ \mathbb{P}\left\{ #1 \right\} }

\newcommand{\Bern}{{\rm Bern}}

\newcommand{\Pois}{{\rm Pois}}

\newcommand{\ie}{i.e.\xspace}

\newcommand{\pth}[1]{\left( #1 \right)}
\newcommand{\qth}[1]{\left[ #1 \right]}
\newcommand{\sth}[1]{\left\{ #1 \right\}}

\newcommand{\iprod}[2]{\left \langle #1, #2 \right\rangle}
\newcommand{\Iprod}[2]{\langle #1, #2 \rangle}
\newcommand{\indc}[1]{{\mathds{1}_{\left\{{#1}\right\}}}}

\newcommand{\calB}{{\mathcal{B}}}

\newcommand{\calN}{{\mathcal{N}}}

\newcommand{\calX}{{\mathcal{X}}}

\DeclareMathAlphabet{\varmathbb}{U}{bbold}{m}{n}

\newcommand{\argmax}{\mathrm{argmax}}
\newcommand{\argmin}{\mathrm{argmin}}

\renewcommand{\hat}{\widehat}
\renewcommand{\tilde}{\widetilde}

\usepackage{bbm}

\newcommand{\ER}{Erd\H{o}s-R\'{e}nyi\xspace}

\hypersetup{
            CJKbookmarks=true,
            bookmarksnumbered=true,
            bookmarksopen=true,
            colorlinks=true,
            citecolor=red,
            linkcolor=blue,
            anchorcolor=red,
            urlcolor=blue
}

\title{Consistent recovery threshold of hidden nearest neighbor graphs \\
}

\author{Jian Ding, Yihong Wu, Jiaming Xu, and Dana Yang\thanks{
J.\ Ding is with Department of Statistics, The Wharton School, University of Pennsylvania, Philadelphia, USA, \texttt{dingjian@wharton.upenn.edu}.
Y.\ Wu is with Department of Statistics and Data Science, Yale University, New Haven, USA, \texttt{yihong.wu@yale.edu}.
J.\ Xu and D.\ Yang are with The Fuqua School of Business, Duke University, Durham NC, USA, \texttt{\{jx77,xiaoqian.yang\}@duke.edu}.
}}
\date{\today}

\begin{document}

\maketitle

\begin{abstract}

Motivated by applications such as discovering strong ties in social networks
and assembling genome subsequences in biology, we study the problem of recovering a hidden $2k$-nearest neighbor (NN) graph in an $n$-vertex complete graph,
whose edge weights are independent and distributed according to $P_n$ for edges in the hidden $2k$-NN graph 
and $Q_n$ otherwise. The special case of Bernoulli distributions corresponds to a variant of the Watts-Strogatz small-world graph.
We focus on two types of asymptotic recovery guarantees as $n\to \infty$: 
(1) exact recovery: all edges are classified correctly with probability tending to one;
(2) almost exact recovery: the expected number of misclassified edges is $o(nk)$. 
We show that the maximum likelihood estimator achieves (1) exact recovery
for $2 \le k \le n^{o(1)}$ if $ \liminf \frac{2\alpha_n}{\log n}>1$; 
(2) almost exact recovery for $ 1 \le k \le o\left( \frac{\log n}{\log \log n} \right)$ 
if 
$
 \liminf \frac{kD(P_n||Q_n)}{\log n}>1,
$
where $\alpha_n \triangleq -2 \log \int \sqrt{d P_n d Q_n}$ is the R\'enyi divergence of order $\frac{1}{2}$
and $D(P_n||Q_n)$ is the Kullback-Leibler divergence. 
Under mild distributional assumptions, these conditions are shown to be information-theoretically necessary for any algorithm to succeed. 
 A key challenge in the analysis is the enumeration of $2k$-NN graphs that
 differ from the hidden one by a given number of edges.

\end{abstract}

\section{Introduction}
\label{sec:intro}

The strong and weak ties are essential for information diffusion, social cohesion, and community organization in social networks~\cite{granovetter1977strength}.
The strong ties  between close friends are responsible for forming 
tightly-knit groups,
while the weak ties  between acquaintances are crucial for binding groups of strong ties together~\cite{easley2010networks}. 
The celebrated Watts-Strogatz small-world graph~\cite{watts1998collective} is a simple network model that exhibits both strong and weak ties. It posits that $n$ nodes are located on a ring and
starts with a $2k$-nearest neighbor (NN) graph of strong ties, where each node is connected to its $2k$ nearest neighbors ($k$ on the left and $k$ on the right) on the ring. Then to generate weak ties, for every node, each of its strong ties is rewired with probability $\epsilon$
to a node chosen uniformly at random. As $\epsilon$ varies from $0$ to $1$, the  graph interpolates
between a ring lattice and an \ER random graph; for intermediate values of $\epsilon$,
the graph is a small-world network: highly clustered with many triangles,
yet with a small diameter. 

The Watts-Strogatz small-world graph and its variants, albeit 
simple, have been extensively studied and widely used in various disciplines to model real networks beyond social networks, such as
academic collaboration network~\cite{newman2001structure}, metabolic networks~\cite{wagner2001small}, 
brain networks~\cite{bassett2006small}, and word co-occurrence networks in language modeling~\cite{cancho2001small,
motter2002topology}.  Most of the previous work
focuses on studying the structures of small-world graphs~\cite{newman1999scaling} and their algorithmic consequences~\cite{Kleinberg00,moore2000epidemics,saramaki2005modelling}.
However, in many practical applications, it is also of interest to distinguish strong ties 
from weak ones~\cite{marsden1984measuring,gilbert2009predicting,gilbert2012predicting,rotabi2017detecting}. For example, in Facebook~\cite{facebook} or Twitter network~\cite{FM2317}, identifying the close ties among a user's potentially 
hundreds of friends provides valuable information for marketing and ad placements. 
Even when additional link attribute information (such as the communication time in who-talks-to-whom networks~\cite{onnela2007structure}) are available to be used to measure the strength of the tie, the task of discovering strong ties could still be challenging, as the link attributes are potentially noisy or only partially observed, obscuring the inherent tie strength. Therefore, it is of fundamental importance, in both theory and practice, to understand
when and how we can infer strong ties from the noisy and partially observed network data. 
In this paper, we address this question in the following statistical model:
\begin{definition}[Hidden $2k$-NN graph recovery] \label{def:2knnhidden}\hfill

\noindent
{\bf Given}: $n \ge 1$, and two distributions $P_n$ and $Q_n$, parametrized by $n$.\\
{\bf Observation}: A randomly weighted, undirected complete graph $w$
 with a hidden $2k$-NN graph $x^*$ on $n$ vertices, such that 
for each edge $e$, the edge weight $w_e$ is distributed as $P_n$ if $e$ is an edge in $x^*$ and as $Q_n$ otherwise. \\
{\bf Inference Problem:} Recover the hidden $2k$-NN graph  $x^*$ from the observed random graph.
\end{definition}

\begin{figure}[ht]%
        \centering
        \begin{subfigure}{0.3\textwidth}
            \centering
            \includegraphics[width=\textwidth]{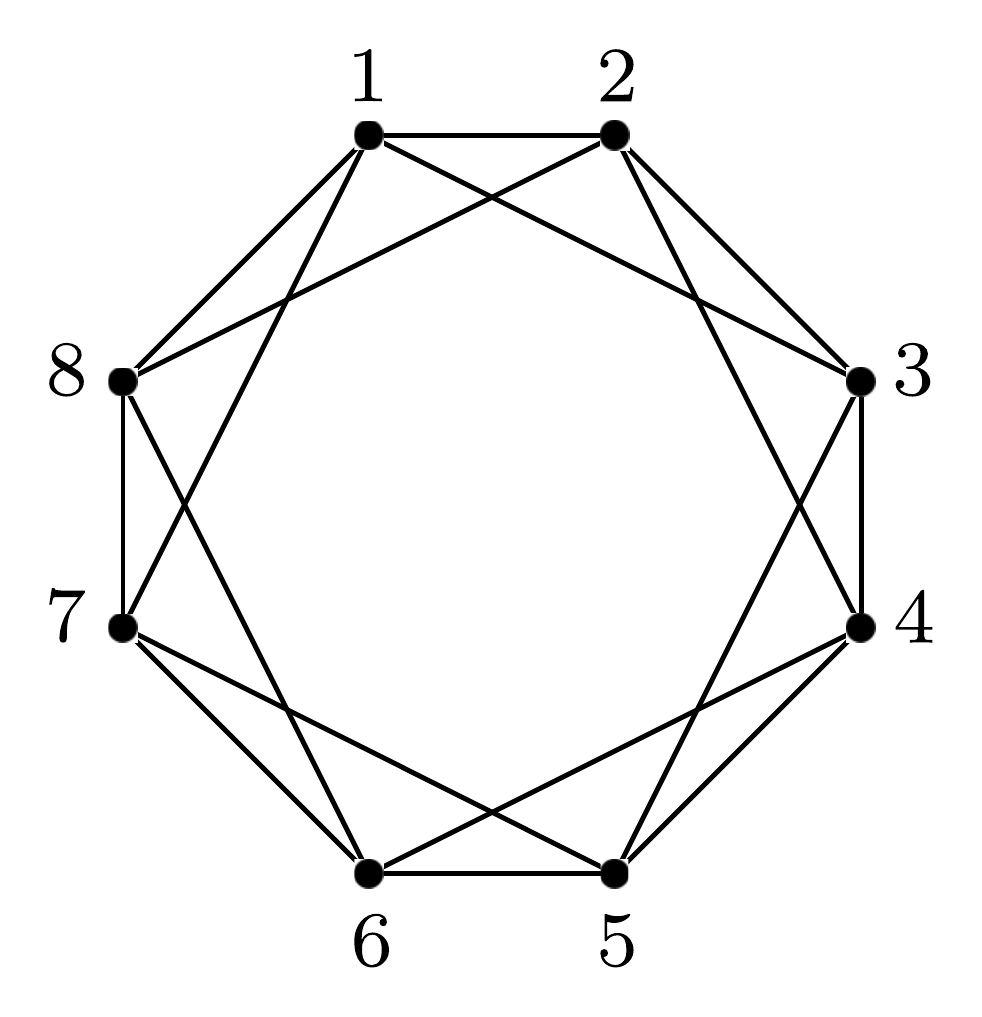} 
            \caption{The $2k$-NN graph corresponding to the Hamiltonian cycle $(1,2,3,4,5,6,7,8,1)$.}
            \label{fig:2kNN.a}
        \end{subfigure}
        \quad
        \begin{subfigure}{0.3\textwidth}  
            \centering 
            \includegraphics[height=\textwidth]{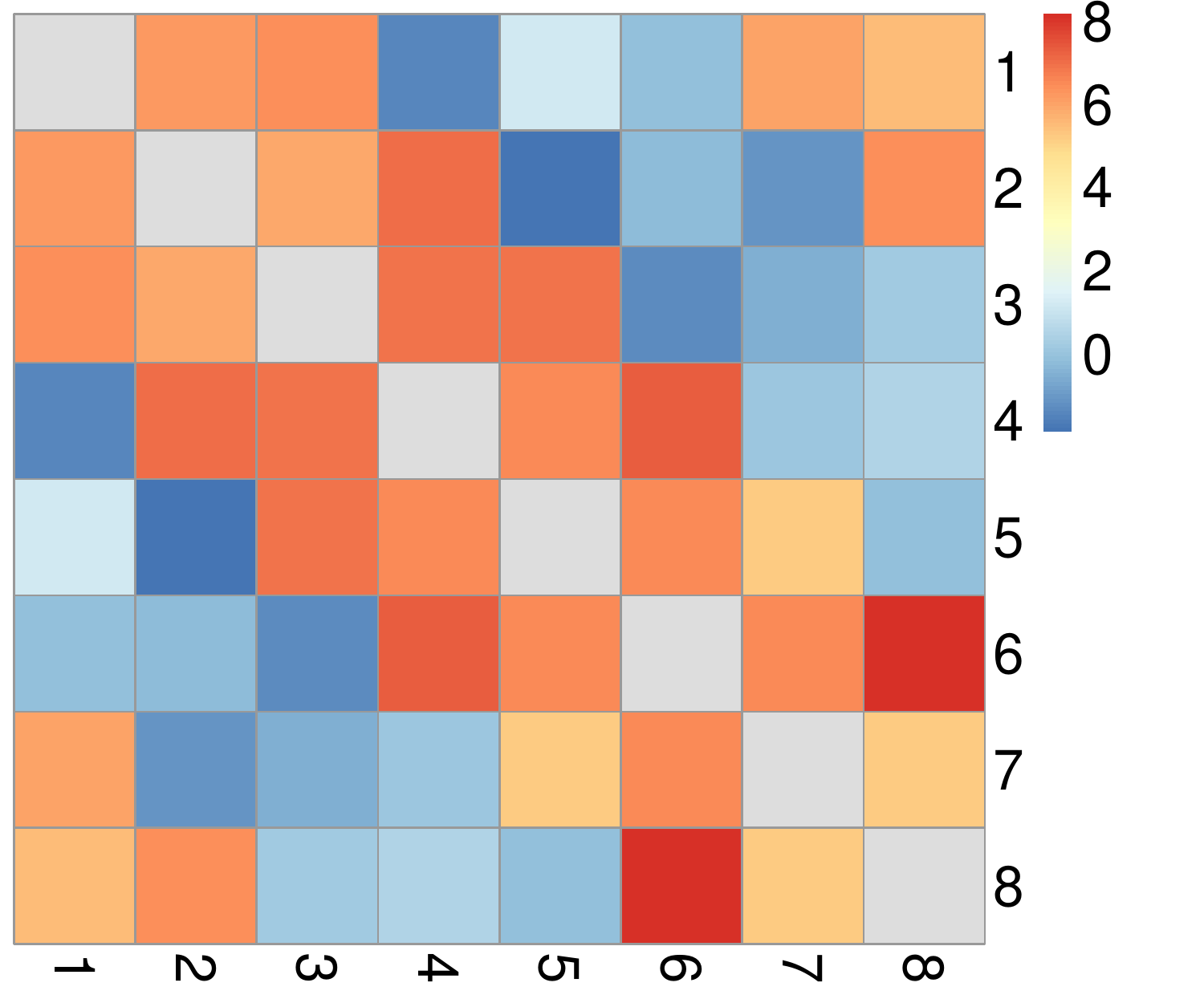} 
            \caption{Heatmap of one realization of $w$ with the hidden $2k$-NN graph in (a).}
            \label{fig:2kNN.b}
        \end{subfigure}
        \vskip\baselineskip
        \begin{subfigure}{0.3\textwidth}   
            \centering 
            \includegraphics[width=\textwidth]{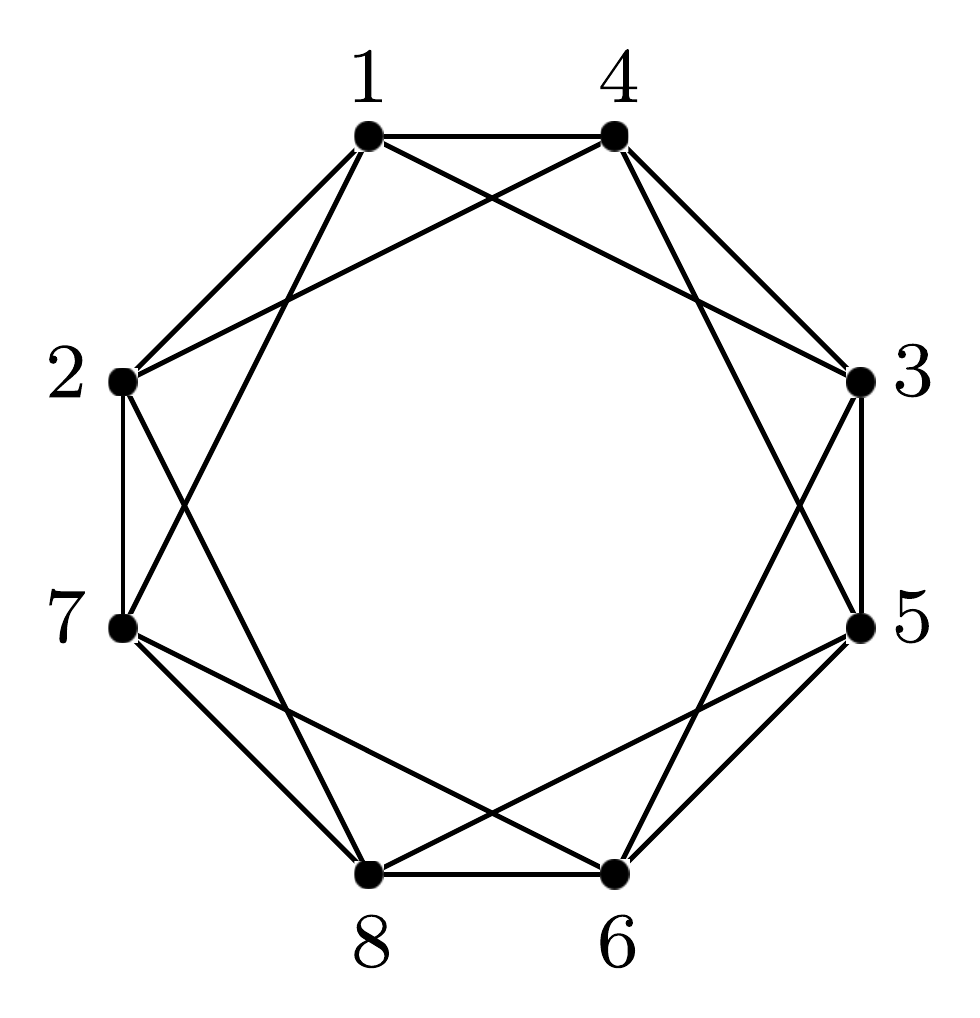}
            \caption{The $2k$-NN graph corresponding to the Hamiltonian cycle $(1,4,3,5,6,8,7,2,1)$.}
            \label{fig:2kNN.c}
        \end{subfigure}
        \quad
        \begin{subfigure}{0.3\textwidth}   
            \centering 
            \includegraphics[height=\textwidth]{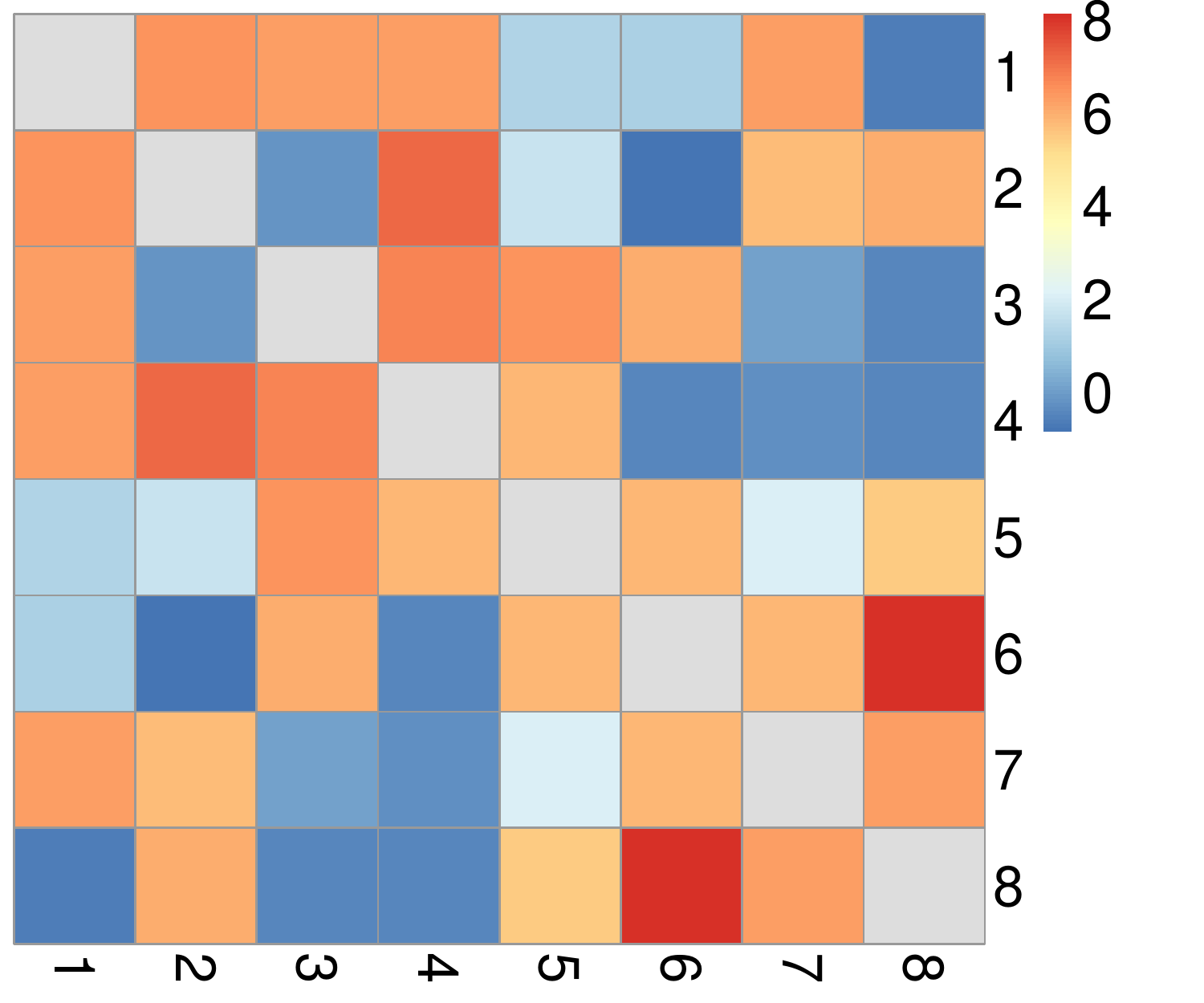}  
            \caption{Heatmap of one realization of $w$ with the hidden $2k$-NN graph in (c).}
            \label{fig:2kNN.d}
        \end{subfigure}
\caption{Examples of $2k$-NN graphs for $n=8$ and $k=2$. The edge weight $w_e$ is distributed $P_n=\mathcal{N}(6, 1)$ if $e$ is an edge in the $2k$-NN graph. Otherwise $w_e\sim Q_n=\mathcal{N}(0,1)$.
}
\label{fig:2knn}%
\end{figure}

\begin{figure}[ht]%
      \centering
      \begin{subfigure}{0.4\textwidth}
      \centering
       \includegraphics[width=.9\textwidth]{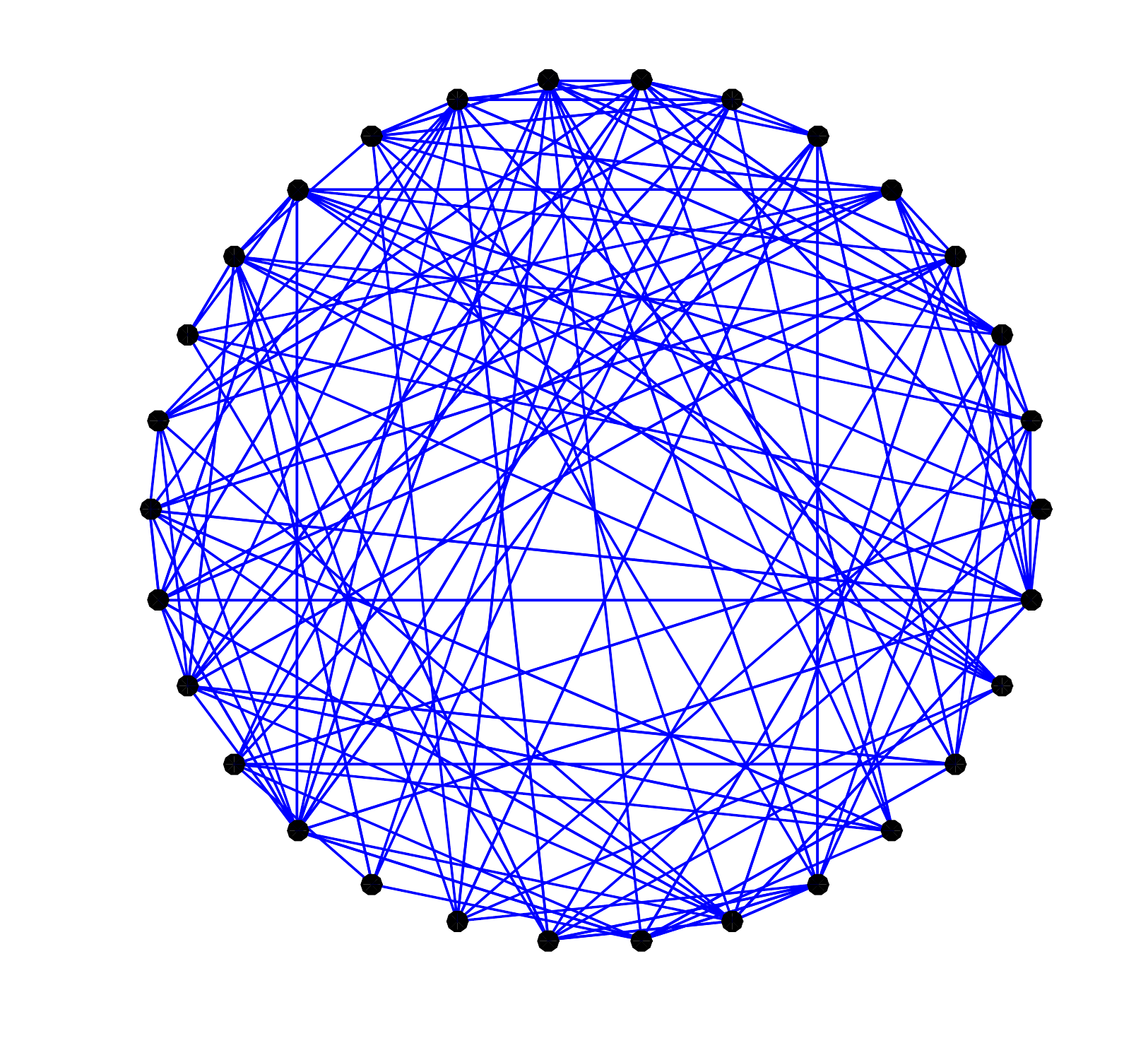}
       \end{subfigure}
       \hskip -4ex
      { \Large $\xrightarrow[\text{recovery}]{\text{$2k$-NN}}  $}
      \hskip -2ex
      \begin{subfigure}{.4\textwidth}
      \centering
       \includegraphics[width=.9\textwidth]{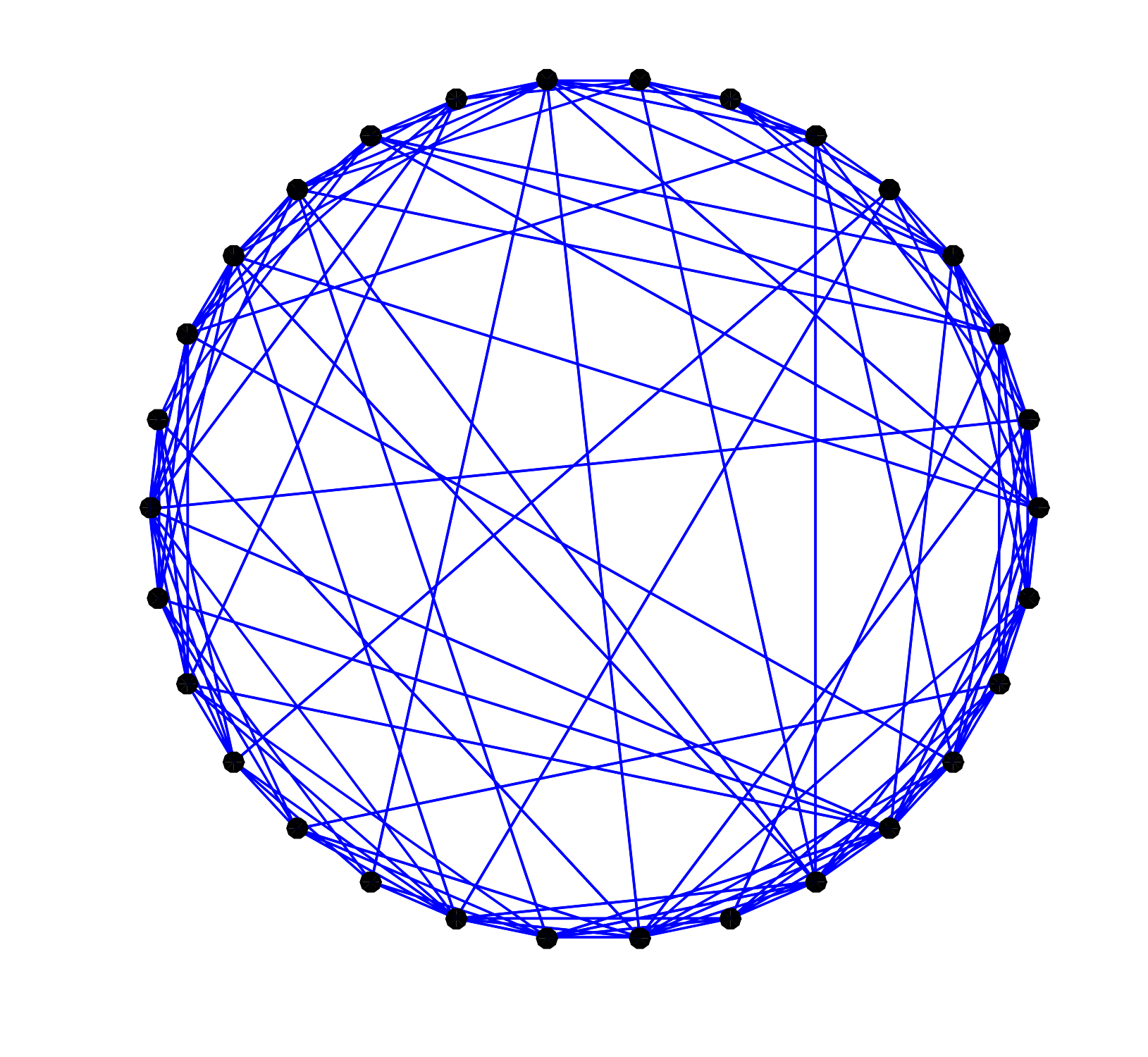}
       \end{subfigure}
       \caption{Left: An observed graph generated by the hidden $2k$-NN graph model 
       with $n=30$ vertices, $k=4$, $P_n=\Bern(0.8)$, and $Q_n=\Bern(0.09)$;
       Right: the observed graph with vertices rearranged according to the latent $2k$-NN graph.}
       \label{fig:small_world}
\end{figure}


See Figs.~\ref{fig:2knn} and~\ref{fig:small_world} for graphical illustrations of the model and the reconstruction problem.
Note that every $2k$-NN graph $x$ can be described by a permutation $\sigma$ on $[n]$ as follows: first, construct a Hamiltonian cycle $(\sigma(1), \sigma(2), \ldots, \sigma(n), \sigma(1))$, then connect pairs of vertices that are at distance at most $k$ on the cycle (cf.~Figs.~\ref{fig:2kNN.a} and \ref{fig:2kNN.c}). 

Our model encompasses the case of partially observed networks. This can be accomplished by considering $P_n = \epsilon \delta_* + (1-\epsilon) P_n'$ and $Q_n = \epsilon \delta_* + (1-\epsilon) Q_n'$ where $*$ is a special symbol outside of the support of $P_n'$ and $Q_n'$ indicating those edge weights that are unobserved. 
When $P_n$ and $Q_n$ are Bernoulli distributions with corresponding success probabilities $p_n>q_n$, we arrive at a variant of the Watts-Strogatz small-world graph.  

The problem of recovering a hidden NN graph is also motivated by 
\emph{de novo genome assembly}, the reconstruction of an organism's
long sequence of $A,G,C,T$ nucleotides from fragmented sequencing data,
which is one of the most pressing challenges in genomics~\cite{gnerre2011high,chapman2011meraculous}. 
A key obstacle 
of the current high-throughput genome assembly technology is 
\emph{genome scaffolding}, that is, 
extending genome subsequences (so-called contigs) to the whole genome by ordering them according to their positions on the genome.
Thanks to recent
advances in sequencing technology~\cite{lieberman2009comprehensive, putnam2016chromosome}, this process is aided by long-range linking information between contigs in the form of randomly sampled Hi-C reads,
where a much larger concentration of Hi-C reads exist between nearby contigs on the genome than those that are far apart.
By representing each contig as a node, the underlying true ordering of contigs on the genome 
as the hidden Hamiltonian cycle, and the counts of the Hi-C reads linking the contigs
as edge weights, the previous work~\cite{bagaria2018hidden} casts genome scaffolding as a hidden Hamiltonian cycle recovery problem with $P_n = \Pois(\lambda_n)$ and $Q_n = \Pois(\mu_n)$ with $\lambda_n \ge \mu_n$, where $\lambda_n$ and $\mu_n$ are the average number of Hi-C reads between adjacent and non-adjacent contigs respectively; 
this is a special case of our model for $k=1$. 
However, this hidden Hamiltonian cycle model only takes into account of the signal -- an elevated mean number of Hi-C reads -- in the immediately adjacent contigs on the genome; in reality, nearby contigs (e.g.~two-hop neighbors) also demonstrate stronger signal than those that are far apart. By considering $k>1$, our general $2k$-NN graph model is a closer approximation to the real data, capturing the large Hi-C counts observed between near contigs which can be used to better assemble the genome. 
Indeed, as our theory later suggests, the information provided by multi-hop neighbors strictly improves the recovery threshold.

Note that in the aforementioned applications we often have $k \ll n$; thus in this 
paper we focus on the regime of $k=n^{o(1)}$ and study the following two types of recovery guarantees. 
Let $x^* \in \{0,1\}^{\binom{n}{2}}$ denote the adjacency vector of the hidden $2k$-NN graph, where
$x^*_e=1$ for every edge $e$ in the hidden $2k$-NN graph and $x^*_e=0$ otherwise.
Let $\calX$ denote the collection of adjacency vectors of all $2k$-NN graphs with vertex set $[n]$.

\begin{definition}[Exact recovery]
An estimator $\widehat{x}=\hat{x}(w) \in \{0,1\}^{\binom{n}{2}}$ 
achieves exact recovery if, as $n\to\infty$,
\[
\sup_{x^* \in \calX}\mathbb{P}\left\{ \hat{x} \neq x^*\right \}=o(1).
\]
where $w$ is distributed according to the hidden $2k$-NN graph model in \prettyref{def:2knnhidden} with hidden $2k$-NN graph $x^*$. 
\end{definition}

Depending on the applications, we may not be able to reconstruct the hidden $2k$-NN graph $x^*$ perfectly;
instead, we may consider correctly estimating all but a small number of edges, which is required to be $o(nk)$, since a $2k$-NN graph contains $kn$ edges. 
In particular, let $d(x, \widehat{x})$ be the Hamming distance
$$
d(x^*, \widehat{x}) = \sum_{e\in\binom{[n]}{2}} \indc{ x^*_e\neq \hat{x}_e}.
$$


\begin{definition}[Almost exact recovery]
An estimator $\widehat{x}=\hat{x}(w) \in \{0,1\}^{\binom{n}{2}}$ achieves almost exact recovery if, as $n\to \infty$, 
\[
\sup_{x^* \in \calX}\expect{d( x^*, \hat{x} )}=o(nk).
\]  
\end{definition}

Intuitively, for a fixed network size $n$ and a fixed number $k$ of nearest neighbors,
as the distributions $P_n$ and $Q_n$ get closer, the recovery problem becomes harder. 
This leads to an immediate question: 
\emph{From an information-theoretic perspective, computational considerations aside, what are the
fundamental limits of recovering the hidden $2k$-NN graph?}
To answer this question, we derive necessary and sufficient conditions in terms of
the model parameters $(n,k,P_n, Q_n)$ under which the hidden $2k$-NN graph can be exactly
or almost exactly recovered. These results serve as benchmarks
for evaluating practical algorithms and aid us in understanding
the performance limits of polynomial-time algorithms.

Specifically, we discover that the following two information measures characterize the sharp thresholds for exact and almost exact recovery,
respectively. Define the R\'enyi divergence of order $1/2$:\footnote{It is also related to the so-called 
Battacharyya distance $B(P_n,Q_n)$ via $\alpha_n = 2 B(P_n,Q_n)$.}
\begin{equation}
\alpha_n=-2\log \int \sqrt{\diff P_n \diff Q_n};
\label{eq:alpha}
\end{equation}
and the Kullback-Leibler divergence:
\[
D(P_n\| Q_n)=\int \diff P_n \log \frac{\diff P_n}{\diff Q_n}.
\]
Under some mild assumptions on $P_n$ and $Q_n$, 
we show that the necessary and sufficient conditions are as follows:
\begin{itemize}
\item Exact recovery ($2 \le k \le n^{o(1)}$): 
\begin{equation}
\liminf_{n\rightarrow \infty}\frac{2\alpha_n}{\log n}>1;
\label{eq:exact-threshold}
\end{equation}
\item Almost exact recovery $ \left( 1 \le k \le o\left( \frac{\log n}{\log \log n} \right) \right)$:
\begin{equation}
\liminf_{n\rightarrow\infty} \frac{kD(P_n||Q_n)}{\log n}>1.
\label{eq:approx-threshold}
\end{equation}
\end{itemize}
The conditions for exact recovery and almost exact recovery are characterized by two different distance measures $\alpha_n$ and $D(P_n\| Q_n)$. This comes from the large deviation analysis for different regimes of $d(x^*,\hat{x})$. See remark~\ref{rmk:measures} for a detailed explanation.
For the special case of $k=1$ (Hamiltonian cycle), the exact recovery condition was shown to be
$\liminf_{n\rightarrow \infty}\frac{\alpha_n}{\log n}>1$~\cite{bagaria2018hidden}.
Comparing this with \prettyref{eq:exact-threshold} for $k\ge 2$,
we find that, somewhat surprisingly, the exact recovery threshold is halved when $k$ increases from $1$ to $2$,
and then stays unchanged as long as $k$ remains $n^{o(1)}$.
In contrast, the almost exact recovery threshold decreases inversely proportional to $k$ over the range of $[1, o(\log n/\log\log(n))]$. 
The sharp thresholds of exact recovery for $k \ge n^{\Omega(1)}$ and 
almost exact recovery for $k=\Omega(\log n/\log \log n)$ are open.

For the Bernoulli distribution (in other words, unweighted graphs) with $P_n=\Bern(p)$ and $Q_n=\Bern(q)$,
we have the explicit expressions of 
 $$
 \alpha_n=- 2 \log \left(\sqrt{pq} + \sqrt{(1-p)(1-q)} \right)  \quad \text{ and } \quad
 D(P_n\| Q_n)=p \log \frac{p}{q} + (1-p) \log \frac{1-p}{1-q}.
 $$
As an interesting special case, consider the parametrization
\begin{equation}
p=1-\epsilon + \frac{2\epsilon k}{n-1} \quad \text{and}  \quad q= \frac{2\epsilon k}{n-1},
\label{eq:smallworld}
\end{equation}
so that the mean number of edges in the observed graph stays at  $nk$ for all $\epsilon \in [0,1]$. 
This can be viewed as an approximate version of the Watts-Strogatz small-world graph, 
in which we start with a $2k$-NN graph, then rewire each edge with probability $\epsilon$
independently at random. In this case, our main results specialize to:
\begin{itemize} 
      \item Exact recovery  is possible if and only if 
\begin{equation}
\epsilon=
\begin{cases}
o(1/n)  &  k=1 \\
o(1/\sqrt{n}) & 2 \le k \le n^{o(1)}
\end{cases}
\label{eq:ws-exact}
\end{equation}
\item Almost exact recovery is possible if and only if 
\begin{equation}
k(1-\epsilon)\ge 1+o(1) \quad \text{ for } 1 \le k \le o(\log n/\log \log n).
\label{eq:ws-almost}
\end{equation}
\end{itemize}

In the related work~\cite{cai2017detection}, a similar case of Bernoulli distributions  has been studied.\footnote{To be precise,
the previous work~\cite{cai2017detection} considers Bernoulli distributions 
under a slightly different parameterization:
 $p=1-\epsilon + \frac{2\epsilon^2 k}{n-1}$ and $q=\frac{2\epsilon k}{n-1}$. 
In addition to exact recovery and approximate recovery, 
 a hypothesis testing problem between the small-world graph and \ER random graph 
 is studied. 
} It is shown in~\cite{cai2017detection} that exact recovery is impossible if
$1-\epsilon=o\left( \sqrt{ \frac{\log n}{n}} \vee \frac{\log n}{k} \frac{1}{\log\frac{n\log n}{k^2}} \right)$.
In particular, this impossibility result requires $\epsilon \to 1$, which 
is highly suboptimal compared to the sharp exact recovery condition \prettyref{eq:ws-exact}. 
It is also shown in~\cite{cai2017detection} that almost exact recovery
can be achieved efficiently via thresholding on the number of common neighbors when
$1-\epsilon=\omega\left( (\frac{\log n}{n})^{1/4} \vee (\frac{\log n}{k})^{1/2} \right)$
and via spectral ordering when $1-\epsilon= \omega \left(\frac{n^{3.5}}{k^4}\right)$;
these sufficient conditions, however, are very far from being optimal.

Finally, we remark that our sharp exact and almost exact recovery thresholds are achieved by
the maximum  likelihood estimator (MLE)  for the hidden $2k$-NN graph problem, which is computationally
intractable in the worst case. 
For the special case $k=1$, the exact recovery threshold 
is shown to be achieved efficiently in polynomial-time via
a linear programming (LP) relaxation of the MLE (namely, the fractional 2-factor LP)~\cite{bagaria2018hidden}; moreover, we show that 
the almost exact recovery threshold can be achieved efficiently in polynomial-time
via a simple thresholding procedure. For $k\ge 2$, however, it remains open whether the exact recovery threshold or the almost exact recovery threshold
can be achieved efficiently in polynomial-time (see~\prettyref{sec:discussion}
for a more detailed discussion).

\section{Exact recovery}\label{sec:exact}

	
	

The maximum likelihood estimator  for the hidden $2k$-NN graph problem is equivalent to finding the \emph{max-weighted $2k$-NN subgraph} with weights given by the log likelihood ratios.
Specifically, assuming that $dP_n/dQ_n$ is well-defined, for each edge $e \in \binom{[n]}{2}$, let $L_e=\log \frac{\diff P_n}{\diff Q_n}(w_e)$. Then the MLE is the solution to the following combinatorial optimization problem:
\begin{align}
\widehat{x}_{\rm ML} = \arg \max_{x \in \calX}  & \;  \iprod{ L} {x},   \label{eq:mle} 
\end{align}
where we recall that $\calX$ denotes the collection of adjacency vectors of all $2k$-NN graphs on $[n]$.
Note that in the Poisson, Gaussian or Bernoulli 
model where the log likelihood ratio is an affine function of the edge weight, we can simply replace $L$ in \prettyref{eq:mle} by the edge weights $w$.

Recall that $\alpha_n=-2\log \int \sqrt{dP_n dQ_n}$. We show that if $2\leq k\leq n^{o(1)}$, then the condition $\liminf_{n\rightarrow\infty} (2\alpha_n/\log n)>1$ is sufficient for $\hat{x}_{\text{ML}}$ to achieve exact recovery. This condition is also necessary, with the following additional assumption (which is fulfilled by a wide class of weight distributions including Poisson, Gaussian and Bernoulli distributions \cite{bagaria2018hidden}):  


\begin{assumption}[{\cite[Lemma~1]{bagaria2018hidden}}] 
\label{assump:exact.lower}
Let $X=\log \frac{\diff P_n}{\diff Q_n}(\omega_x)$ for some $\omega_x\sim P_n$ and $Y=\log\frac{\diff P_n}{\diff Q_n}(\omega_y)$ for some $\omega_y\sim Q_n$. Assume that
\[
\sup_{\tau\in\mathbb{R}}\left(\log \mathbb{P}\left\{Y\geq \tau\right\}+\log \mathbb{P}\left\{X\leq \tau\right\}\right)\geq -(1+o(1))\alpha_n+o(\log n).
\]
\end{assumption}

The following is our main result regarding exact recovery.

\begin{theorem}[Exact recovery]
\label{thm:exact}
Let $k \ge 2$.
\begin{itemize}
	\item Suppose 
\begin{align}
\alpha_n-\frac{1}{2} (\log n+17\log k) \to \infty. \label{eq:exact_recovery_suff}
\end{align}
Then the MLE \prettyref{eq:mle} achieves exact recovery: $\mathbb{P}\left\{\hat{x}_{\mathrm{ML}}\neq x^*\right\} \to 0$. In particular, this holds if $k=n^{o(1)}$ and 
\[
\liminf_{n\rightarrow \infty}\frac{2\alpha_n}{\log n}>1.
\]
 
\item
Conversely, assume that $k< n/12$ and \prettyref{assump:exact.lower} holds. If exact recovery is possible, then
\[
\liminf_{n\rightarrow \infty}\frac{2\alpha_n}{\log n}\geq 1.
\]
\end{itemize}
\end{theorem}

When $k=1$, as shown in~\cite{bagaria2018hidden} the sharp threshold for exact recovery is $\liminf_{n\rightarrow \infty} \frac{\alpha_n}{\log n}>1$, which is 
stronger than the condition in Theorem~\ref{thm:exact} by a factor of 2. In other words, from $k=1$ to $k\geq 2$ there is a strict decrease in the required level of signal. 
A simple explanation is that the hidden
$2k$-NN graph $x^*$ contains more edges when $k\geq 2$,
and the elevated weights on these edges provide extra signal for determining the latent permutation $\sigma^*$.
However, this extra information ceases to help as $k$ increases from $2$ to $n^{o(1)}$, which can be attributed to the following fact: when we swap any pair of adjacent vertices on $\sigma^*$, we always get a $2k$-NN graph $x$ which differ from $x^*$ by $4$ edges,
regardless of how large $k$ is. 
This argument results in the $k$-independent necessary condition $\liminf_{n\rightarrow \infty}\frac{2\alpha_n}{\log n}\ge1$ for the MLE to succeed
 (see~\prettyref{sec:exact.lower} for details).

In the introduction, we have discussed the implications of this result when the edge weights are distributed as Bernoulli.
\prettyref{thm:exact} can be applied to a wide range of continuous and discrete weight distributions.
In particular, our result implies that when $2\leq k\leq n^{o(1)}$,
\begin{itemize}
\item for $P_n=\mathcal{N}(\mu_n,1)$ and $Q_n=\mathcal{N}(\nu_n,1)$, the sharp threshold for exact recovery is
\[
\liminf_{n\rightarrow \infty}\frac{(\mu_n-\nu_n)^2}{2\log n}>1;
\]
\item for $P_n=\Pois(\mu_n)$ and $Q_n=\Pois(\nu_n)$, the sharp threshold for exact recovery is
\[
\liminf_{n\rightarrow\infty}\frac{2\left(\sqrt{\mu_n}-\sqrt{\nu_n}\right)^2}{\log n}>1.
\]
\end{itemize}

\subsection{Proof of correctness of MLE for exact recovery}
\label{sec:exact-suff}

To analyze the MLE, we first introduce the notion of \emph{difference graph}, which encodes the difference between a proposed $2k$-NN graph and the ground truth.
Given $x,x^*\in \{0,1\}^{{n\choose 2}}$, let $G=G(x)$ be a bi-colored simple graph on $[n]$ whose adjacency vector is $x-x^*\in \{0,\pm 1\}^{{n\choose 2}}$, in the sense that 
each pair $(i,j)$ is connected by a blue (resp.~red) edge if $x_{ij}-x^*_{ij} = 1$ (resp.~$-1$).
See \prettyref{fig:diff} for an example.
By definition, red edges in $G(x)$ are true edges in $x^*$ that are missed by the proposed solution $x$, and blue edges correspond to spurious edges that are absent in the ground truth.

\begin{figure}[H]%
\centering
\includegraphics[width=0.3\columnwidth]{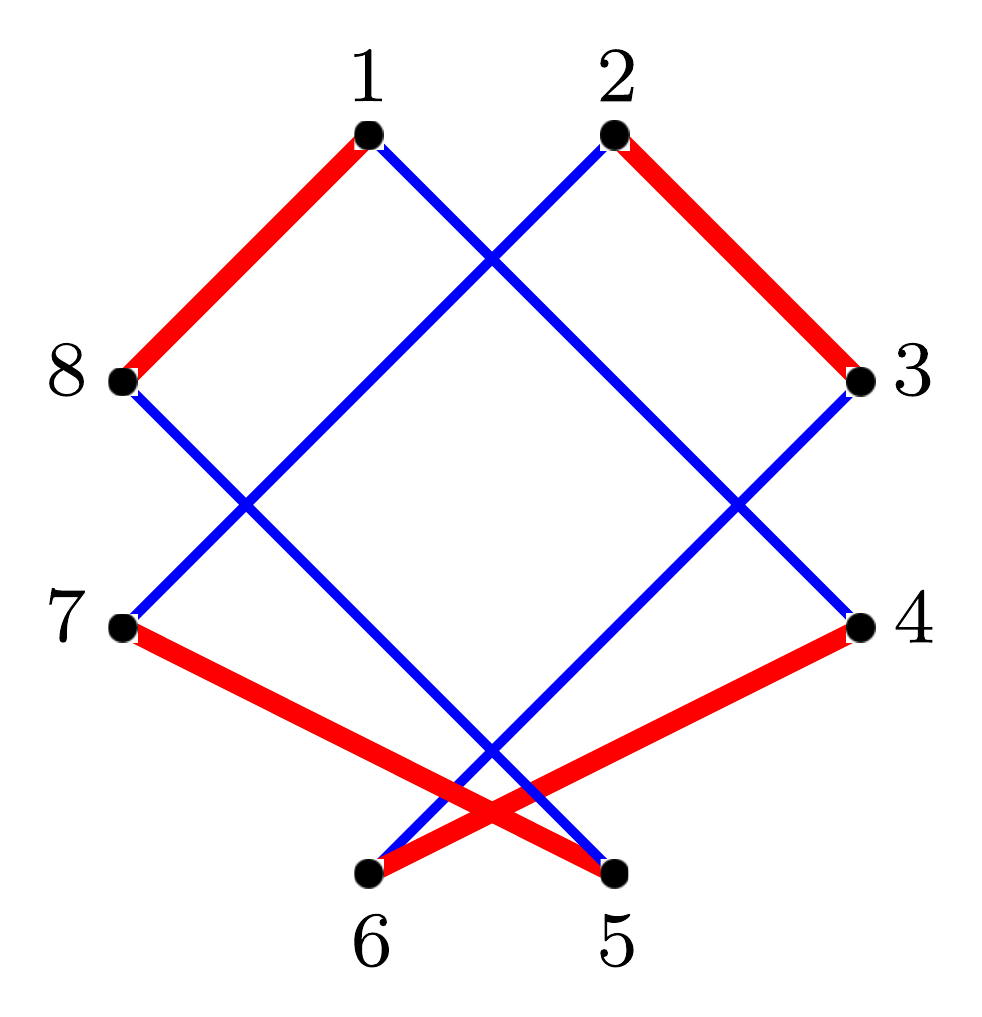}%
\caption{An example for a difference graph $G$. Here $G$ is obtained by letting $x^*$ (resp.~$x$) be the $2k$-NN graph in Fig.~\ref{fig:2kNN.a} (resp.~\ref{fig:2kNN.c}), and then taking the difference $x-x^*$. The red (thick) edges stand for edges that in $x^*$ but not $x$, while the blue (thin) edges are in $x$ but not $x^*$.}
\label{fig:diff}%
\end{figure}

A key property of difference graphs is the following: 
Since $2k$-NN graphs are $2k$-regular, the difference graph $G$ is \emph{balanced}, in the sense that for each vertex, its red degree (the number of incident red edges) coincides with its blue degree. Consequently, $G$ has equal number of red edges and blue edges, and the number of red (or blue) edges measures the closeness of $x$ to the truth $x^*$. 
Denote
\begin{equation}
\mathcal{X}_\Delta = \left\{x\in\mathcal{X}: d(x, x^*)=2\Delta\right\}=\left\{x\in \mathcal{X}: G(x)\text{ contains exactly }\Delta \text{ red edges}\right\}.
\label{eq:XDelta}
\end{equation}
In particular, $\{\calX_\Delta: \Delta \geq 0\}$ partitions the feasible set $\calX$.
The analysis of the MLE relies crucially on bounding the size of $\mathcal{X}_\Delta$.
To see this, note that by the definition of the MLE, 
\begin{align}
\mathbb{P}\left\{ \hat{x}_\text{ML} \neq x^* \right\}
& \le \mathbb{P}\left\{\exists x\in \mathcal{X}: \langle L, x-x^*\rangle \ge 0 \right\} \nonumber \\
& \overset{(a)}{\leq}  \sum_{\Delta\geq 1}  \sum_{x \in \mathcal{X}_\Delta } \mathbb{P}\left\{\langle L,x-x^*\rangle \ge 0\right\} 
\nonumber \\
&  \overset{(b)}{\leq} \sum_{\Delta\geq 1}  \left| \calX_\Delta \right| \exp\left(-\alpha_n \Delta\right),
\label{eq:MLE_intuition}
\end{align}
where $(a)$ follows from the union bound and $(b)$ follows from the Chernoff bound. It remains to derive a tight bound
to the size of $\mathcal{X}_\Delta$, which turns out to be much more challenging when $k\ge 2$ than $k=1$.  

To provide some intuitions, let us first prove a simple bound 
\begin{equation}
|\mathcal{X}_\Delta| \leq (4kn)^\Delta,
\label{eq:Xdeltacard-simple}
\end{equation}
 by following similar arguments as in \cite[Sec.~4.2]{bagaria2018hidden} for analyzing the MLE for $k=1$ (Hamiltonian cycles). Substituting \prettyref{eq:Xdeltacard-simple} into 
\prettyref{eq:MLE_intuition} immediately yields that $\mathbb{P}\left\{ \hat{x}_\text{ML} \neq x^* \right\} \to 0$,
provided that $\alpha_n- \log (nk) \to +\infty$, which falls short of the desired sufficient condition~\eqref{eq:exact_recovery_suff} by roughly a factor of $2$ when $k\ge2$.

For each $x\in\calX_\Delta$, suppose its difference graph $G$ consists of $m$ connected components $G_1,...,G_m$. Then each connected component is also a balanced bi-colored graph.
Let $\Delta_i\geq 1$ denote the number of red edges in $G_i$. There are at most $2^\Delta$ configurations for the sequence $(\Delta_1,...,\Delta_m)$
since $\sum_{i\leq m}\Delta_i=\Delta$. 
From~\cite[Lemma~1]{bagaria2018hidden}, every connected balanced bi-colored graph has an \emph{alternating Eulerian circuit},
{\it i.e.}~a circuit with colors alternating between red and blue that passes through every edge exact once.
To bound the total number of configurations for a connected component $G_i$ with $\Delta_i$ red edges,
it suffices to count the number of such alternating Eulerian circuits, which is upper bounded by
the number of length-$2\Delta_i$ path $(v_0,v_1,..., v_{2\Delta_i-1})$, such that
$(v_i,v_{i+1})$ is a red edge if $i$ is even, and a blue edge if $i$ is odd. To complete the circuit, 
$(v_{2\Delta_i-1},v_0)$
must be a blue edge.

We now sequentially enumerate $v_0$ to $v_{2\Delta_i-1}$: given $v_0$, which takes $n$ values, 
there are only $2k$ possibilities for $v_1$, because $(v_0,v_1)$ is a red edge so that $v_1$ must belong to the neighborhood of 
$v_0$ in the true $2k$-regular graph $x^*$.
Overall, we conclude that the path $(v_0,\ldots,v_{2\Delta_i-1})$ can take at most $(2kn)^{\Delta_i}$ possible values.
Summing over the connected components, we have
\begin{align}
\left|\mathcal{X}_\Delta \right|\leq \sum_{(\Delta_1,...,\Delta_m): \sum \Delta_i=\Delta} \left(\prod_{i\leq m}(2kn)^{\Delta_i}\right)\leq 2^\Delta(2kn)^\Delta=(4kn)^\Delta. \label{eq:Xdeltacard_simple}
\end{align}
It turns out this bound is only tight for $k=1$. 
For $k\ge 2$, the following combinatorial lemma (proved in \prettyref{sec:counting}) gives a much better bound on 
the cardinality of $\mathcal{X}_\Delta$, which, together with
\prettyref{eq:MLE_intuition}, immediately leads to the desired sufficient condition~\eqref{eq:exact_recovery_suff}.
\begin{lemma}
\label{lmm:Xdeltacard}	
There exists an absolute constant $C$ such that for any $\Delta \geq 0$ and any $2 \leq k \leq n$ 
	\begin{equation}
|\calX_\Delta| \leq 2 \left(C k^{17}n\right)^{\Delta/2}.	
	\label{eq:Xdeltacard}
	\end{equation}
\end{lemma}
In comparison with the simple bound \prettyref{eq:Xdeltacard_simple}, 
\prettyref{lmm:Xdeltacard} improves the dependency on $n$ from $n^\Delta$ to $n^{\Delta/2}$.
We have already seen that the red edges play an important role in the proof of \prettyref{eq:Xdeltacard_simple}, 
as the entropy of red edges is much lower than that of blue edges. 
The improved bound \prettyref{eq:Xdeltacard} 
is obtained by further exploiting the structural properties of red edges in the difference graph $G$. 
In particular, we find that for each red edge in $G$, there is at least another red edge ``close'' to it, which allows us to count red edges in groups and further reduce their entropy.
 The precise notion of closeness 
will be given in \prettyref{sec:counting}, but let us illustrate with a simple example.

\begin{example}
\sloppy
Recall that each $2k$-NN graph $x$ can be identified with a permutation
$(\sigma(1), \sigma(2), \ldots, \sigma(n))$. By connecting adjacent nodes on $\sigma$ and connecting $\sigma(n)$ with $\sigma(1)$, $\sigma$ determines a Hamiltonian cycle, from which one can connect pairs of vertices whose distance is at most $k$ to construct a $2k$-NN graph. Suppose the true $2k$-NN graph $x^*$ is identified with the identity permutation $\sigma^* = (1,2,...,n)$. Consider the alternative graph $x$ identified
with a permutation that traverses part of the vertices in the opposite direction, {\it i.e.} $\sigma= (1,2,..., i, j, j-1,...,i+1, j+1,j+2,...,n)$ for some $i,j$ that are far apart (see~\prettyref{fig:ex_perm}). The corresponding difference graphs in the $k=1$ and $k=2$ cases are illustrated in~\prettyref{fig:ex_kis1},~\ref{fig:ex_kis2} respectively. 

\begin{figure}[H]
    \centering
    \begin{subfigure}[b]{0.3\textwidth}
    \begin{centering}
        \includegraphics[height=.8\linewidth]{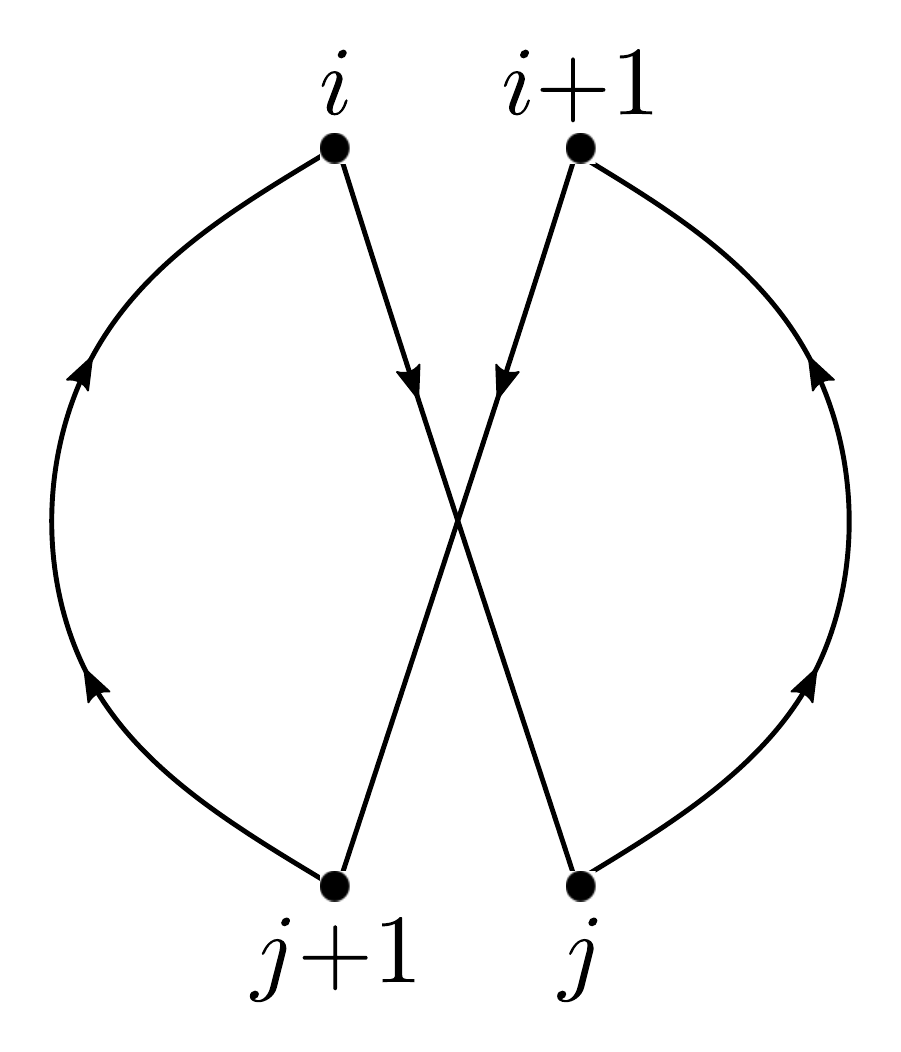}
        \caption[justification   = centering]{}
        \label{fig:ex_perm}
     \end{centering}
    \end{subfigure}
    \begin{subfigure}[b]{0.32\textwidth}
    \begin{centering}
        \includegraphics[height=.8\linewidth]{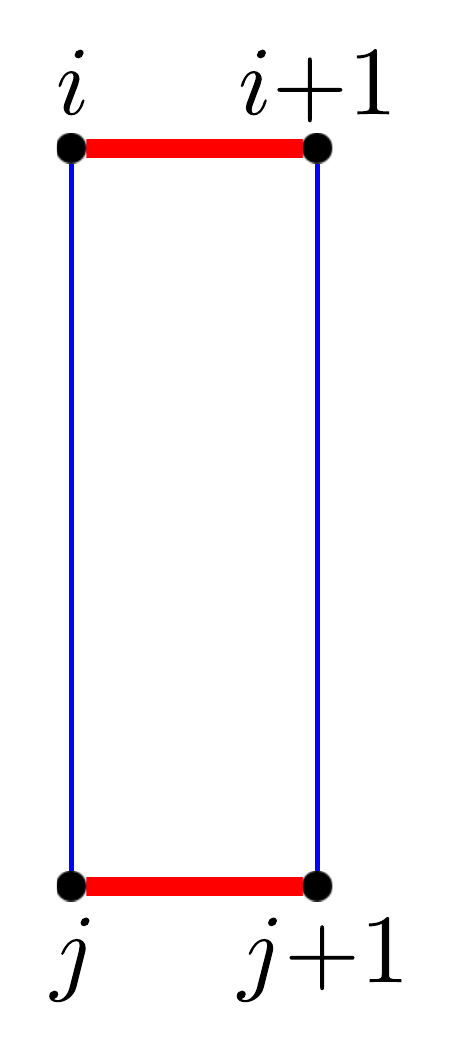}
        \caption[justification   = centering]{}
        \label{fig:ex_kis1}
    \end{centering}
    \end{subfigure}
    \begin{subfigure}[b]{0.35\textwidth}
    \begin{centering}
        \includegraphics[height=.8\linewidth]{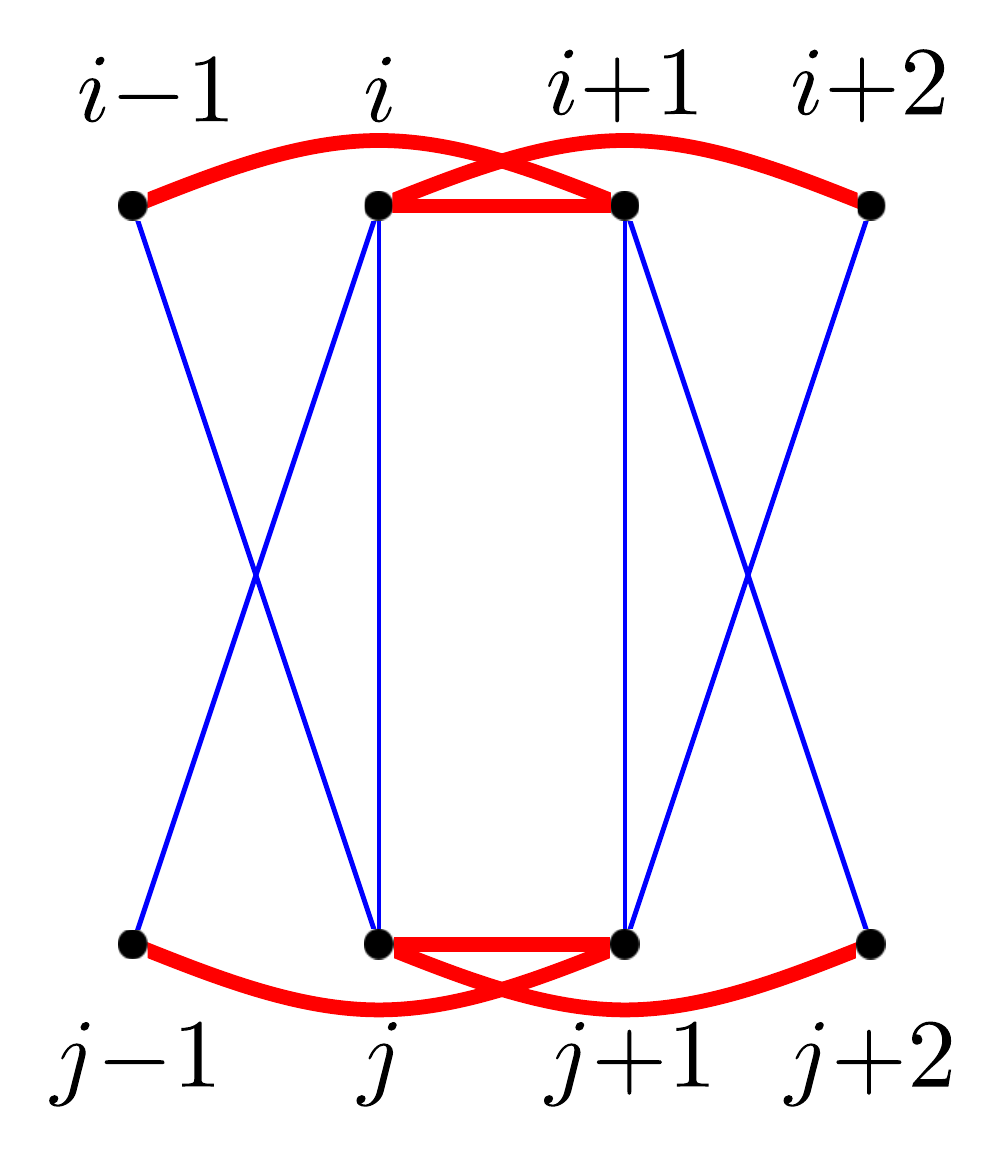}
        \caption[justification   = centering]{}
        \label{fig:ex_kis2}
    \end{centering}
    \end{subfigure}
    \caption{From left to right, (a): Hamiltonian cycle associated with the permutation $\sigma= (...,i,j,j-1,...,i+1,j+1,j+2,...)$; (b): the difference graph $G(x)$ when $k=1$, where $x$ is the $2k$-NN graph identified by $\sigma$; (c): the difference graph $G(x)$ when $k=2$.}
    \label{fig:ex_ij}
\end{figure}

The crucial observation is that when $k\geq 2$, there is more structure in the set of red edges in the sense that red edges do not appear in isolation. For example in~\prettyref{fig:ex_kis2}, the indices of the three red edges $(i,i+1), (i-1,i+1)$ and $(i,i+2)$ are all close to each other; in particular, this triple can only take $n$ values in total. We find that when $k\geq 2$, this observation holds in greater generality. As a result, 
each red edge can help determine at least one other red edge, allowing us to enumerate the red edges in bundles. This is the main reason why when upper bounding $|\mathcal{X}_\Delta|$, we can reduce the exponent on $n$ from $\Delta$ to $\Delta/2$.
This structural property, however, 
is specific to $k\geq 2$: 
as shown by \prettyref{fig:ex_kis1} ($\Delta=2$), the simple bound \prettyref{eq:Xdeltacard_simple} is tight for $k=1$.

\end{example}

Let us also point out that the exponent $\Delta/2$ in \prettyref{eq:Xdeltacard} is tight for $k\geq 2$. It is easy to see that when the nodes $i,i+1$ in the permutation $\sigma^*$ are swapped, a difference graph $G(x)$ is formed with $\Delta=2$ red edges (see \prettyref{sec:exact.lower} and \prettyref{fig:lowerbound}  for a more extensive discussion of this example). Taking $i=1,...,n$ yields $n$ distinct difference graphs that are all members of $\mathcal{X}_2$, meaning the size of $\mathcal{X}_\Delta$ is no smaller than $n^{\Delta/2}$, at least when $\Delta=2$.

Assuming \prettyref{lmm:Xdeltacard}, the proof of the correctness of $\hat{x}_{\text{ML}}$ follows from the Chernoff bound and the union bound. 
\begin{proof}[Proof of sufficiency part of \prettyref{thm:exact}] 
First partition $\mathcal{X}$ according to the value of $\Delta$:
\begin{equation}\label{eq:stratify}
\mathbb{P}\left\{\exists x\in \mathcal{X}: \langle L, x-x^*\rangle \ge 0 \right\}
\leq \sum_{\Delta\geq 1}\mathbb{P}\left\{\exists x\in \mathcal{X}_\Delta: \langle L,x-x^*\rangle \ge 0\right\}.
\end{equation}
Recall that $L_e=\log \frac{\diff P_n}{\diff Q_n}(w_e)$. Hence for each $x\in\mathcal{X}_\Delta$,
the law of $\langle L,x-x^*\rangle$ only depends on $\Delta$, which can be represented as follows:
\[
\langle L,x-x^*\rangle \stackrel{d}{=}\sum_{i\leq \Delta}Y_i-\sum_{i\leq \Delta} X_i,
\]
where $X_1,\ldots, X_\Delta$ are i.i.d.~copies of $\log \frac{\diff P_n}{\diff Q_n}$ under $P_n$, $Y_1,\ldots, Y_\Delta$ are i.i.d. copies of $\log \frac{\diff P_n}{\diff Q_n}$ under $Q_n$, and $\stackrel{d}{=}$ denotes equality in distribution. 
Applying the Chernoff bound yields
\begin{equation}
\mathbb{P}\left\{\sum_{i\leq \Delta}Y_i-\sum_{i\leq \Delta}X_i \ge 0\right\}
\leq \inf_{\lambda>0} \left\{ \exp\left(\Delta\left(\psi_Q(\lambda)+\psi_P(-\lambda)\right)\right) \right\},
\label{eq:chernoff-XY}
\end{equation}
where 
\begin{align}
 \psi_Q(\lambda) & \triangleq \log \Expect_{Q_n}\qth{\exp\pth{\lambda \log \frac{\diff P_n}{\diff Q_n}}}
= \log \int \diff P_n^{\lambda} \diff Q_n^{1-\lambda}, \label{eq:psiQ}\\
 \psi_P(\lambda) & \triangleq \log \Expect_{P_n}\qth{\exp\pth{\lambda \log \frac{\diff P_n}{\diff Q_n}}} 
= \log \int \diff P_n^{1+\lambda} \diff Q_n^{-\lambda}  
= \psi_Q(\lambda + 1 ), \label{eq:psiP}
\end{align}
denote the log moment generating functions (MGFs) of $\log \frac{\diff P_n}{\diff Q_n}$ under $P_n$ and $Q_n$, respectively.
In particular, the R\'enyi divergence in \prettyref{eq:approx-threshold} is given by
\begin{equation}
\alpha_n  = -2\psi_Q\left(\frac{1}{2}\right) = -2\psi_P\left(-\frac{1}{2}\right).
\label{eq:alpha-psi}
\end{equation}
Choosing $\lambda=1/2$ in \prettyref{eq:chernoff-XY} yields
\begin{equation}
\prob{\langle L,x-x^*\rangle \ge 0} =
\mathbb{P}\left\{\sum_{i\leq \Delta}Y_i-\sum_{i\leq \Delta}X_i \ge 0\right\}
\leq \exp\left(2\Delta\psi_Q\left(\frac{1}{2}\right)\right)=\exp\left(-\alpha_n \Delta\right).
\label{eq:chern}
\end{equation}
Combining \prettyref{eq:Xdeltacard} and \prettyref{eq:chern} and applying the union bound, we have
\begin{align*}
\mathbb{P}\left\{\exists x\in \mathcal{X}_\Delta: \langle L,x-x^*\rangle \ge 0\right\}\leq \sum_{x\in \mathcal{X}_\Delta}\mathbb{P}\left\{\langle L,x-x^*\rangle \ge 0\right\} \leq 2\exp\left(-\Delta \kappa_n \right),
\end{align*}
where $\kappa_n \triangleq \alpha_n-\frac{\log (Ck^{17}n)}{2} \to \infty$ by assumption. 
Finally, from \prettyref{eq:stratify},
\begin{align*}
\mathbb{P}\left\{\exists x\in \mathcal{X}: \langle L, x-x^*\rangle \ge 0 \right\}
\leq & \sum_{\Delta\geq 1}2\exp\left(-\Delta \kappa_n\right) = \frac{2\exp(-2\kappa_n) }{1-\exp(-\kappa_n)} \xrightarrow{n\diverge} 0.
\end{align*}
In other words, $\mathbb{P}\{\hat{x}_\text{ML}=x^*\}\rightarrow 1$ as $n\rightarrow\infty$.
\end{proof}

\subsection{Information-theoretic lower bound for exact recovery}\label{sec:exact.lower}
For the purpose of lower bound, consider the Bayesian setting where $x^*$ is drawn uniformly at random from the set $\calX$ of all $2k$-NN graphs. Then MLE maximizes the probability of success, which, by definition, can be written as follows:
\begin{equation}\label{eq:bayes.risk}
\mathbb{P}\left\{\widehat{x}_{\text{ML}}=x^*\right\}=\mathbb{P}\left\{\langle L,x-x^*\rangle<0,\;\;\; \forall x\neq x^*\right\}.
\end{equation}
Due to the symmetry of $\mathcal{X}$, the probabilities in~\eqref{eq:bayes.risk} are equal to the corresponding conditional probabilities, conditional on each $x^*\in \mathcal{X}$. WLOG, assume that $x^*$ is the $2k$-NN graph associated with the identity permutation $\sigma^*(i)=i$.
It is difficult to work with the intersection of dependent events in \prettyref{eq:bayes.risk}. The proof strategy is to select a subset of feasible solutions for which the events $\langle L,x-x^*\rangle<0$ are mutually independent. 

To this end, define $x^{(i)}$ to be the $2k$-NN graph corresponding to the permutation $\sigma$ that swaps $i$ and $i+1$, i.e., $\sigma(i)=i+1$, $\sigma(i+1)=i$, and $\sigma=\sigma^*$ everywhere else. It is easy to see that the difference graph $G(x^{(i)})$ contains four edges: (see \prettyref{fig:lowerbound})
\begin{align*}
\text{red edges:}\;\;\;& (i-k,i),\;\;\;(i+1,i+k+1);\\
\text{blue edges:}\;\;\;& (i-k,i+1),\;\;\;(i,i+k+1).
\end{align*}
\begin{figure}[H]
\centering
\includegraphics[width = 5 in]{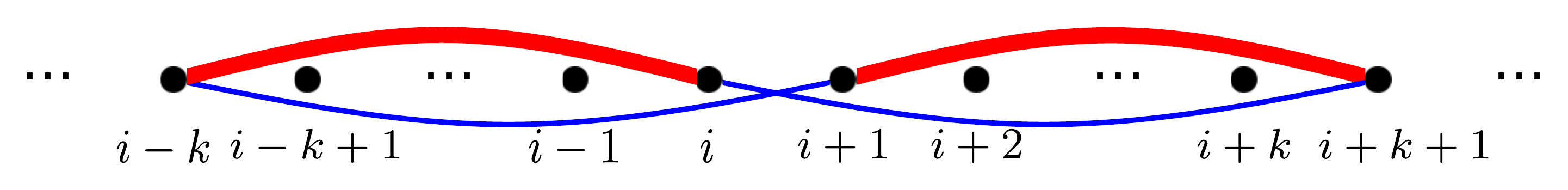}
\caption{The difference graph $G(x^{(i)})$.}
\label{fig:lowerbound}
\end{figure}
Furthermore, for two such graphs $x^{(i)}$ and $x^{(j)}$ with $k+1\leq i< j\leq n-k$, the edges sets $E\left(G(x^{(i)})\right)$ and $E\left(G(x^{(j)})\right)$ intersect if and only if $j-i\in \{k,k+1\}$. To avoid such pairs, we divide the $x^*$ cycle into blocks of $3k$, each further divided into three sections of length $k$, and only consider those $i$ which lies in the middle section of a block. Formally, define
\[
D=\left\{k+1,k+2,\ldots, 2k, 4k+1,\ldots, 5k, \ldots, 3k(\lfloor n/3k\rfloor-1)+k+1,\ldots, 3k(\lfloor n/3k\rfloor-1)+2k \right\}.
\]
Then for distinct $i$ and $j$ in $D$, the difference graph of $x^{(i)}$ and $x^{(j)}$ have disjoint edge sets.
This means all elements of $\left\{\langle L, x^{(i)}-x^*\rangle: i\in D\right\}$ are mutually independent.

For each $i\in D$, we have
\begin{align*}
& \mathbb{P}\{\langle L,x^{(i)}-x^*\rangle<0\} \\
= & \mathbb{P}\left\{L(i-k,i+1)+L(i,i+k+1)-L(i-k,i)-L(i+1,i+k+1)<0\right\}\\
= & \mathbb{P}\left\{Y_1+Y_2-X_1-X_2<0\right\},
\end{align*}
where $X_1$, $X_2$ are independent copies of $\log \frac{\diff P_n}{\diff Q_n}$ under $P_n$, 
and $Y_1$, $Y_2$ are independent copies of $\log \frac{\diff P_n}{\diff Q_n}$ under $Q_n$. Therefore
\begin{align}
\mathbb{P}\left\{\langle L,x-x^*\rangle<0,\;\;\; \forall x\neq x^*\right\}
\nonumber \leq & \mathbb{P}\left\{\langle L, x^{(i)}-x^*\rangle<0, \;\;\; \forall i\in D\right\}\\
\nonumber = & \left(\mathbb{P}\left\{Y_1+Y_2-X_1-X_2<0\right\}\right)^{|D|}\\
\leq & \exp\left(-|D| \mathbb{P}\left\{Y_1+Y_2-X_1-X_2\geq 0\right\}\right). \label{eq:union}
\end{align}
From the mutual independence of $X_1$, $X_2$, $Y_1$, $Y_2$, for any $\tau\in\reals$, we have
\[
\mathbb{P}\left\{Y_1\geq \tau\right\} \mathbb{P}\left\{Y_2\geq \tau\right\} \mathbb{P}\left\{X_1\leq \tau\right\} \mathbb{P}\left\{X_2\leq \tau\right\} \leq 
\mathbb{P}\left\{Y_1+Y_2-X_1-X_2\geq 0\right\}
\]
 and hence
\[
\log \mathbb{P}\left\{Y_1+Y_2-X_1-X_2\geq 0\right\}\geq 2\sup_{\tau\in\mathbb{R}}\left(\log \mathbb{P}\left\{Y_1\geq \tau\right\}+\log \mathbb{P}\left\{X_1\leq \tau\right\}\right).
\]
Since \eqref{eq:union} is an upper bound for $\mathbb{P}\{\widehat{x}_{\text{ML}}=x^*\}$, the success of the MLE must require that
\[
\log |D|+2\sup_{\tau\in\mathbb{R}}\left(\log \mathbb{P}\left\{Y_1\geq \tau\right\}+\log \mathbb{P}\left\{X_1\leq \tau\right\}\right)\rightarrow -\infty.
\]
On one hand, by \prettyref{assump:exact.lower},
\[
\sup_{\tau\in\mathbb{R}}\left(\log \mathbb{P}\left\{Y_1\geq \tau\right\}+\log \mathbb{P}\left\{X_1\leq \tau\right\}\right)\geq -(1+o(1))\alpha_n+o(\log n).
\]
On the other hand, by construction we have $|D| \geq n/3 -k\geq n/4$ under the assumption $k <n/12$, from which 
we conclude the necessity of $2\alpha_n\geq (1+o(1)) \log n$ for $\mathbb{P}\{\widehat{x}_{\text{ML}}=x^*\} \to1$.
%

\subsection{Counting difference graphs}
	\label{sec:counting}

To prove \prettyref{lmm:Xdeltacard}, we begin with some notations. For a $2k$-NN graph $x$, let $E_{\Red}(x)$ (resp.~$E_{\Blue}(x)$) be the set of red (resp.~blue) edges in $G(x)$. The proof strategy is as follows: First, in \prettyref{lmm:red} we count
\[
\mathcal{E}_{\Red}(\Delta)=\left\{E_{\Red}(x):x\in\mathcal{X}_\Delta \right\}. 
\]
Then for each $E_{\Red}\in\mathcal{E}_{\Red}(\Delta)$, \prettyref{lmm:blue} enumerates
\[
\mathcal{X}(E_{\Red}) = \left\{x\in \mathcal{X}_\Delta: E_{\Red}(x)=E_{\Red}\right\}.
\]
which contains all sets of blue edges that are compatible with $E_{\Red}$. This completely specifies the difference graph $G(x)$, and hence the $2k$-NN graph $x$.


For a given $2k$-NN graph $x$ associated with the permutation $\sigma$, let $\mathcal{N}_{x}(i)$ denote the set of neighbors of $i$ in $x$. Let $d_{x}(i,j)=\min\{|\sigma^{-1}(i)-\sigma^{-1}(j)|, n-|\sigma^{-1}(i)-\sigma^{-1}(j)|\}$, which is the distance between $i$ and $j$ on the Hamiltonian cycle defined by $\sigma$. It is easy to check that $d_{x}$ is a well-defined metric on $[n]$. For the hidden $2k$-NN graph $x^*$, define $\mathcal{N}_{x^*}(\cdot)$ and $d_{x^*}(\cdot,\cdot)$ accordingly. 

\begin{definition}\label{def:closeby}
In the $2k$-NN graph $x^*$, define the distance between two edges $e=(i,\tilde{i})$ and $f=(j,\tilde{j})$ as
\[
d(e,f)=\min\{d_{x^*}(i,j), d_{x^*}(i, \tilde{j}), d_{x^*}(\tilde{i}, j), d_{x^*}(\tilde{i},\tilde{j})\}.
\]
We say $e$ and $f$ are \emph{nearby} if $d(e,f)\leq 2k$.
\end{definition}

Since a $2k$-NN graph has a total of $kn$ edges, the cardinality of $\mathcal{E}_{\Red}(\Delta)$ is at most ${kn\choose \Delta}$. The following lemma provides additional structural information for elements of $\mathcal{E}_{\Red}(\Delta)$ that allows us to improve this trivial bound.

\begin{lemma}\label{lmm:2k}
Suppose $k\ge 2$. 
For each red edge $e$ in the difference graph $G$, there exists a nearby red edge $f$ in $G$ that is distinct from $e$.
\end{lemma}

\begin{proof}
We divide the proof into two cases according to the degree of one of the endpoints of $e=(i,\tilde{i})$, say $i$, in the difference graph.
\begin{figure}[ht]
    \centering
    \begin{subfigure}[b]{0.3\textwidth}
    \begin{centering}
        \includegraphics[height=.8\linewidth]{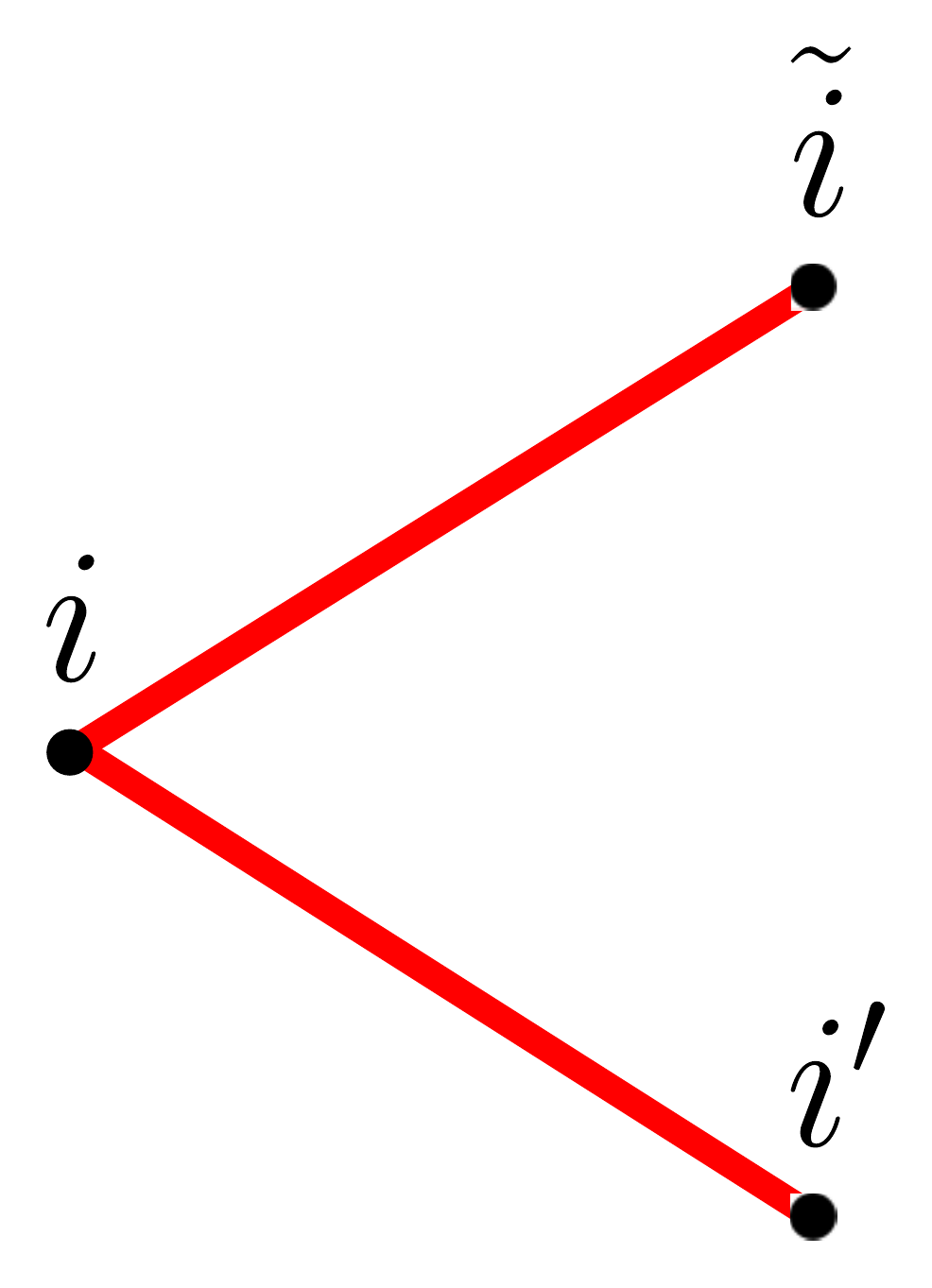}
        \caption[justification   = centering]{Case 1.}
        \label{fig:2k.case.1}
     \end{centering}
    \end{subfigure}
    ~ 
    \begin{subfigure}[b]{0.3\textwidth}
    \begin{centering}
        \includegraphics[height=.8\linewidth]{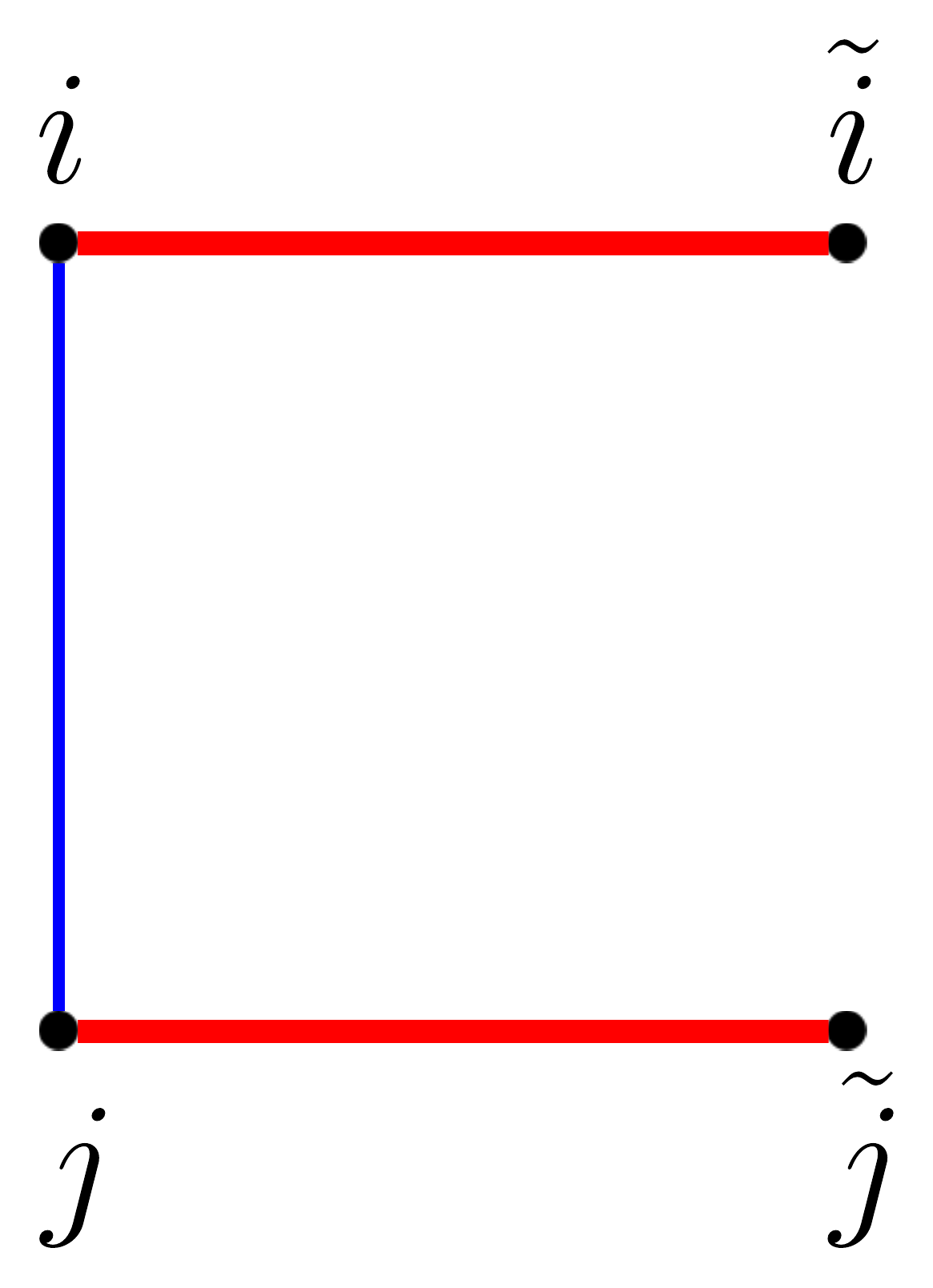}
        \caption[justification   = centering]{Case 2(a).}
        \label{fig:2k.case.2a}
    \end{centering}
    \end{subfigure}
    \hfill
    ~ 
    \begin{subfigure}[b]{0.3\textwidth}
    \begin{centering}
        \includegraphics[height=.8\linewidth]{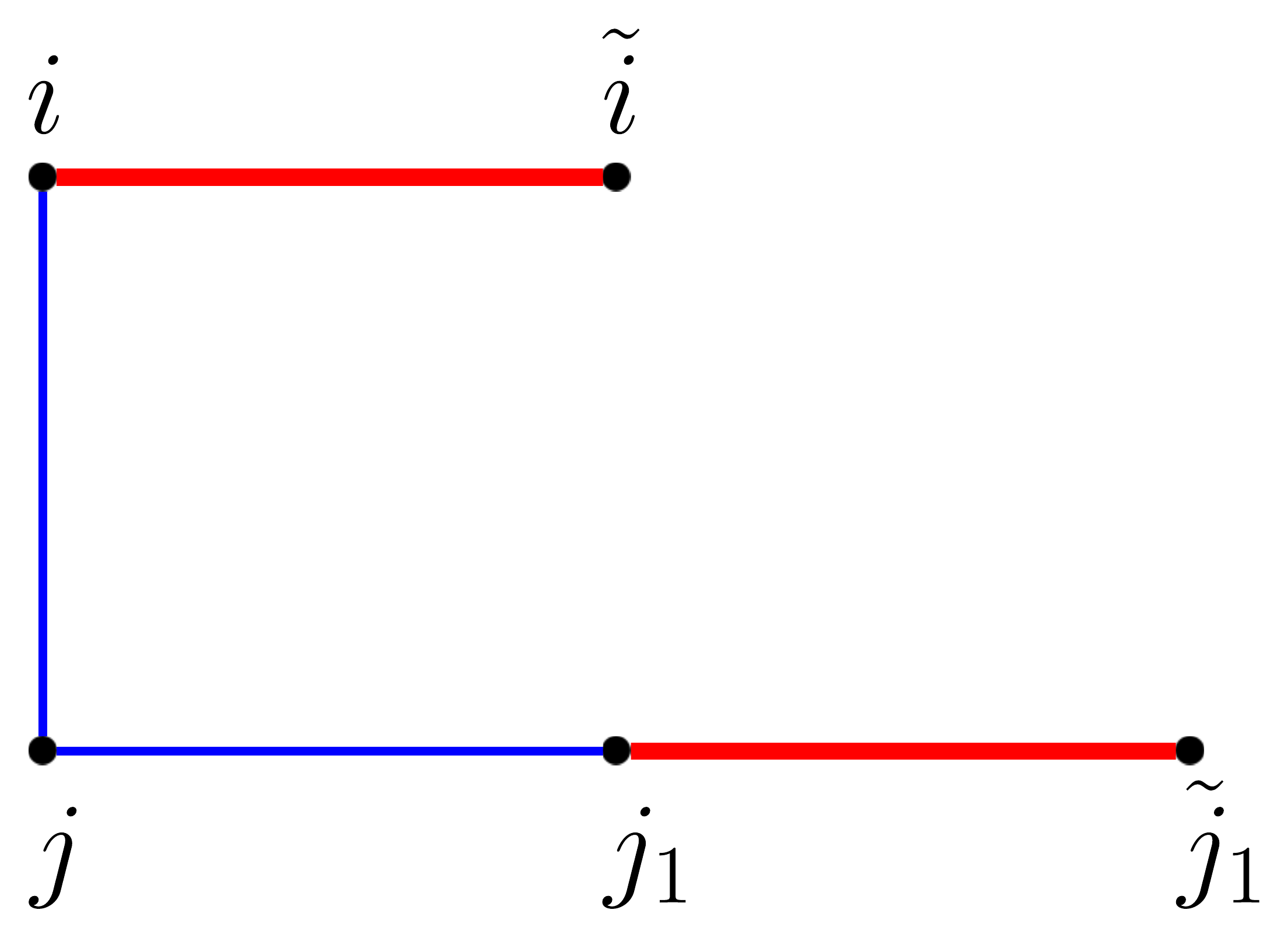}
        \caption[justification   = centering]{Case 2(b).}
        \label{fig:2k.case.2b}
    \end{centering}
    \end{subfigure}
    \caption{Three cases considered in the proof of \prettyref{lmm:2k}.}
    \label{fig:2k}
\end{figure}
\begin{enumerate}
\item 
The degree of $i$ is strictly larger than $2$. Then by balancedness the number of red edges attached to $i$ is at least 2. Other than $(i,\tilde{i})$, there must exist at least one other red edge $(i,i')$. By definition
\[
d((i,\tilde{i}), (i,i'))\leq d_{x^*}(i,i)=0<2k.
\]
That is, $(i,i')$ and $(i,\tilde{i})$ are nearby. See \prettyref{fig:2k.case.1}.
\item 
The degree of $i$ is equal to $2$. Then $i$ is only attached to one red edge and one blue edge in $G$. Denote the blue edge as $(i,j)$. Since the only red edge attached to $i$ is $(i,\tilde{i})$, we have that in the proposed solution $x$, the vertex $i$ is connected to all its old neighbors in $x^*$ except $\tilde{i}$. Thus we get that $\mathcal{N}_x(i)=\mathcal{N}_{x^*}(i)\cup \{j\}\backslash\{\tilde{i}\}$. As a result, 
when $k\geq 2$, out of the two vertices $j_1$, $j_2$ that are right next to $j$ in the $x$ cycle ($d_x(j,j_1)=d_x(j,j_2)=1$), at least one of them is an old neighbor of $i$. WLOG say $j_1\in\mathcal{N}_{x^*}(i)$. Consider these cases:
\begin{enumerate}
\item $d_{x^*}(j,j_1)\leq k$. By triangle inequality $d_{x^*}(j,i)\leq d_{x^*}(j,j_1)+d_{x^*}(i,j_1)\leq 2k$. Because $G$ is a balanced graph, there is at least one red edge $(j,\tilde{j})$ attached to $j$, and
\[
d((i,\tilde{i}), (j, \tilde{j}))\leq d_{x^*}(j,i)\leq 2k.
\]
In other words, $(j,\tilde{j})$ and $(i,\tilde{i})$ are nearby. See \prettyref{fig:2k.case.2a}.

\item $d_{x^*}(j,j_1)>k$. In this case $(j,j_1)$ appears in the difference graph as a blue edge. Therefore $j_1$ is one of the vertices in $G$ and attached to at least one red edge $(j_1,\tilde{j}_1)$. Recall that $j_1\in \mathcal{N}_{x^*}(i)$. Therefore
\[
d((i,\tilde{i}), (j_1,\tilde{j}_1))\leq d_{x^*}(i,j_1)\leq k.
\]
In other words, $(j_1,\tilde{j}_1)$ and $(i,\tilde{i})$ are nearby. See \prettyref{fig:2k.case.2b}.

\end{enumerate}
\end{enumerate}
\end{proof}

The following lemma gives an upper bound for the size of $\mathcal{E}_{\Red}(\Delta)$ and it is a direct consequence of \prettyref{lmm:2k}. 

\begin{lemma}\label{lmm:red}
\[
\left|\mathcal{E}_{\Red}(\Delta)\right|\leq (96k^2)^{\Delta}{kn\choose \lfloor\Delta/2\rfloor}.
\]
\end{lemma}
\begin{proof}
To each member $E_{\Red}$ of $\mathcal{E}_{\Red}(\Delta)$, we associate an undirected graph $\tilde{G}(E_{\Red})$ with vertex set $E_{\Red}$ and edge set $\mathcal{E}(E_{\Red})$, such that for $e,f\in E_{\Red}$, $(e,f)\in\mathcal{E}(E_{\Red})$ if $e$ and $f$ are nearby per \prettyref{def:closeby}. It suffices to enumerate all $E_{\Red}$ for which $\tilde{G}(E_{\Red})$ is compliant with the structural property enforced by \prettyref{lmm:2k}. Our enumeration scheme is as follows:

\begin{enumerate}
\item Fix $m\in[\Delta]$ to be the number of connected components of $\tilde{G}(E_{\Red})$. Select $\{e_1,\ldots, e_m\}$ from the edge set of $x^*$. Since $x^*$ is a $2k$-NN graph with $kn$ edges, there are ${kn\choose m}$ ways to select this set.

\item Let $\Delta_1$, \ldots, $\Delta_m$ be the sizes of the connected components $\mathcal{C}_1,\ldots, \mathcal{C}_m$ of $\tilde{G}(E_{\Red})$. Since $\Delta_i\geq 1$ and $\sum \Delta_i=\Delta$, the total number of such $(\Delta_i)$ sequences is ${\Delta-1 \choose m-1}$, as each sequence can be viewed as the result of replacing $m-1$ of the ``+" symbols with ``," in the expression $\Delta=1+1+\ldots+1+1$. 

\item For each $\mathcal{C}_i$, there is at least one spanning tree $T_i$. Since $\mathcal{C}_i$ and $T_i$ share the same vertex set, it suffices to enumerate $T_i$. First enumerate the isomorphism class of $T_i$, that is, count the total number of unlabeled rooted trees with of $\Delta_i$ vertices. From~\cite{otter1948number}, there are at most $3^{\Delta_i}$ such unlabeled trees. 

\item For $i=1,\ldots, m$, let $e_i$ be the root of $T_i$. Enumerate the ways to label the rest of tree $T_i$. To start, label the vertices on the first layer of $T_i$, that is, the children of $e_i$. A red edge $f$ being a child of $e_i$ on $T_i$ means $f$ and $e_i$ are nearby, limiting the number of labels to at most $16k^2$. To see why, note that at least one endpoint of $f$ is of $d_{x^*}$ distance at most $2k$ from one of the endpoints of $e_i$. No more than $8k$ vertices fit this description. The other endpoint of $f$ can then only choose from $2k$ vertices because $f$ is in the edge set of $x^*$. 

The remaining layers of $T_i$ can be labeled similarly, with at most $16k^2$ possibilities to label each vertex. In total there are at most $(16k^2)^{\Delta_i-1}$ to label $T_i$.

\end{enumerate}

This enumeration scheme accounts for all members of $\mathcal{E}_{\Red}(\Delta)$. By \prettyref{lmm:2k}, $\tilde{G}$ does not contain singletons, {\it i.e.} $\Delta_i\geq 2$ for all $i$. Thus $m\leq \lfloor \Delta/2\rfloor$, and
\begin{align*}
\left|\mathcal{E}_{\Red}(\Delta)\right| \leq& \sum_{m\leq \lfloor \Delta/2\rfloor} {kn\choose m}{\Delta-1\choose m-1}\prod_{i\leq m}3^{\Delta_i}(16k^2)^{\Delta_i-1}\\
\leq& {kn\choose \lfloor\Delta/2\rfloor}2^{\Delta-1}3^\Delta(16k^2)^\Delta\leq (96k^2)^\Delta {kn\choose \lfloor\Delta/2\rfloor}.
\end{align*}

\end{proof}

The following lemma controls the number of $2k$-NN graphs that are compatible with a fixed set of red edges. 
A key observation is that the bound does not depend on $n$.

\begin{lemma}\label{lmm:blue}
For each $E_{\Red}\in\mathcal{E}_{\Red}(\Delta)$,
\begin{align}
\left|\mathcal{X}(E_{\Red})\right|\leq 2(32k^3)^{2\Delta}\Delta^{\Delta/k}. \label{eq:blue_factor}
\end{align}
\end{lemma}

The desired \prettyref{lmm:Xdeltacard} immediately follows from combining \prettyref{lmm:red} and \prettyref{lmm:blue}: 
\begin{align}
|\mathcal{X}_\Delta|
= \left|\bigcup_{E_{\Red}\in\mathcal{E}_{\Red}(\Delta)} \mathcal{X}(E_{\Red}) \right|
 \leq (96k^2)^{\Delta}{kn\choose \lfloor\Delta/2\rfloor}\cdot 2(32k^3)^{2\Delta}\Delta^{\Delta/k}\leq 2\left(C k^{17}n\right)^{\Delta/2} \label{eq:desired_Xdeltacard}
\end{align}
for a universal constant $C>0$, where the last inequality follows from $\binom{a}{b}\le (ea/b)^b$ and $k\ge 2$.
The factor of $\Delta^{\Delta/k}$ in \prettyref{eq:blue_factor} turns out to be crucial. 
To appreciate this subtlety, let us first derive a simple
bound $\left|\mathcal{X}(E_{\Red})\right| \le 4^\Delta \Delta!$. Note that there is a one-to-one correspondence
between $2k$-NN graph $x$ and the difference graph $G(x)$. Hence, it is equivalent to enumerating all possible difference
graphs with the given set of red edges. 
Following the similar alternating Eulerian-circuit based argument for proving \prettyref{eq:Xdeltacard_simple}, we can get that 
$$
\left|\mathcal{X}(E_{\Red})\right| \le  
\sum_{(\Delta_1,...,\Delta_m): \sum \Delta_i=\Delta} \left( 2^\Delta \Delta ! \right) \leq 4^\Delta\Delta!,
$$
where $2^\Delta \Delta !$ counts all the possible orderings of oriented red edges\footnote{To be more precise, 
to count the difference graphs with the given set of $\Delta$ red edges,
it suffices to enumerate all possible edge-disjoint unions of alternating Eulerian circuits with the given set of $\Delta$ red edges.
To this end, for a fixed $m$ and sequence $(\Delta_1, \ldots, \Delta_m)$ such that $\sum \Delta_i=\Delta$,
we enumerate all possible edge-disjoint unions of $m$ alternating Eulerian circuits consisting of $(\Delta_1, \ldots, \Delta_m)$ red edges, respectively. 
First, determine an ordering of oriented red edges, which has 
$2^\Delta \Delta !$ possibilities. Then we connect
the first $\Delta_1$ oriented red edges by blue edges to form the first alternating Eulerian circuit,
the next $\Delta_2$ oriented red edges by blue edges to form the second  alternating Eulerian circuit,
and proceed similarly to form the rest of alternating Eulerian circuits.}. However, this simple
bound  falls short of proving the desired \prettyref{eq:desired_Xdeltacard}, as ${kn\choose \lfloor\Delta/2\rfloor} \Delta !
\ge (c kn)^{\Delta/2} \Delta^{\Delta/2}$ for a universal constant $c>0$. 

\sloppy
\prettyref{lmm:blue} improves over this simple bound by further exploiting the structure in the difference graph. 
In particular, \prettyref{lmm:blue} counts $|\mathcal{X}_\Delta|$ by enumerating all Hamiltonian cycles
$(\sigma(1), \sigma(2), \ldots, \sigma(n),\sigma(1))$ such that $E_{\Red}(x(\sigma))=E_{\Red}$. A key idea is to 
sequentially determine each neighborhood $\mathcal{N}_x\left(\sigma(i)\right)$ starting from $i=1$. 
Suppose $\mathcal{N}_x\left(\sigma(j)\right)$ has been determined for all $1 \le j \le i$
and we are about to specify $\mathcal{N}_x\left(\sigma(i+1)\right)$, which reduces to enumerating $\sigma(i+k+1)$.
Roughly, there are three cases to consider:
\begin{enumerate}
    \item $\sigma(i+1)$ is not in the difference graph $G(x)$. In this case, $\mathcal{N}_x(\sigma(i+1))=\mathcal{N}_{x^*}(\sigma(i+1))$ and thus $\sigma(i+k+1)$ has already been fixed.
    \item $\sigma(i+1)$ is in the difference graph $G(x)$ and $\sigma(i+k+1)$ has fewer than $k$ blue edges
    connecting to $\{\sigma(j): i+1 \le j \le i+k\}$. In this case, at least one of $\{\sigma(j): i+1 \le j \le i+k\}$
    must be a true neighbor of $\sigma(i+k+1)$, which implies that $\sigma(i+k+1)$ has at most $2k^2$ possibilities.
    \item $\sigma(i+1) $ is in the difference graph $G(x)$ and $\sigma(i+k+1)$ has $k$ blue edges
    connecting to $\{\sigma(j): i+1 \le j \le i+k\}$. In this case, $\sigma(i+k+1)$ has at most $2\Delta$ possibilities,
    because the difference graph has at most $2\Delta$ different vertices.
    \end{enumerate}
Note that whenever the last case occur, it gives rise to $k$ new blue edges.
Since the total number of blue edges is $\Delta$,  
the last case can occur at most $\Delta/k$ times, which immediately yields the desired factor $\Delta^{\Delta/k}$ in \prettyref{eq:blue_factor}.

Next, building upon this intuition, we present the rigorous proof of \prettyref{lmm:blue}.

\begin{proof}[Proof of \prettyref{lmm:blue}]
For a given permutation $\sigma$, let $x(\sigma)$ denote the corresponding $2k$-NN graph. Hereafter the dependence on $\sigma$ is suppressed whenever it is clear from the context. These are some useful facts about the difference graph $G$:
\begin{enumerate}
\item Let $V$ denote the collection of the endpoints of edges in $E_\Red$. Then the difference graph $G=(V(G),E(G))$ is given by $V(G)=V$. Since $|E_\Red|=\Delta$, $|V(G)|\leq 2\Delta$.
\item For each $2k$-NN graph, $\sigma$ is determined up to cyclic shifts and a reversals.
\item For two vertices $j\neq j'$, $|\mathcal{N}_x(j)\cap \mathcal{N}_x(j')|$ is completely determined by $d_x(j,j')$ via the formula below.
\[
\left|\mathcal{N}_x\left(j\right)\cap \mathcal{N}_x\left(j'\right)\right|=
\begin{cases}
2k-1-d_x(j,j') & \;\;\;\text{if } d_x(j,j')\leq k;\\
2k+1-d_x(j,j') & \;\;\;\text{if } k<d_x(j,j')\leq 2k;\\
0 & \;\;\; \text{if }d_x(j,j')>2k.
\end{cases}
\]
\end{enumerate}

By fact 2, it suffices to enumerate all $\sigma$ such that $E_{\Red}(x(\sigma))=E_{\Red}$ and $\sigma(1)=1$.  WLOG assume that the ground truth $x^*$, $\sigma^*(i)= i$. The following is the outline of our enumeration scheme:

\begin{enumerate}
\item Enumerate all possibilities for the set $\mathcal{N}_x(1)=\{\sigma(n-k+1), \ldots, \sigma(n), \sigma(2), \ldots, \sigma(k+1)\}$. 
\item With $\mathcal{N}_x(1)$ determined, enumerate all possibilities for the (ordered) sequence $(\sigma(n-k+1), \ldots, \sigma(n), \sigma(2), \ldots, \sigma(k+1))$.
\item For $i$ from $1$ to $n-2k-1$, enumerate $\sigma(i+k+1)$ sequentially, assuming at step $i$ that $\sigma$ were determined from $\sigma(n-k+1)$ up to $\sigma(i+k)$.
\end{enumerate}

Now we give the details on how cardinality bounds are obtained for each step of the enumeration scheme.


{\bf Step 1:} Decompose $\mathcal{N}_x(1)$ according to the set of true neighbors and false neighbors. The set of true neighbors $\mathcal{N}_x(1)\cap \mathcal{N}_{x^*}(1)$ is determined by the set of red edges in $G$. Indeed, this set consists of all members $i\in\mathcal{N}_{x^*}(1)$ for which $(1,i)\notin E_\Red$.

The set $\mathcal{N}_x(1)\backslash \mathcal{N}_{x^*}(1)$ cannot be read directly from the set of red edges. However we know all members of this set must be connected to $1$ via a blue edge. Hence $\mathcal{N}_x(1)\backslash \mathcal{N}_{x^*}(1)$ is a subset of $V(G)$, the vertex set of $G$. Since $V(G)$ is determined by $E_\Red$ and $|V(G)|\leq 2\Delta$, the number of possibilities for $\mathcal{N}_x(1)\backslash \mathcal{N}_{x^*}(1)$ does not exceed the number of subsets of $V(G)$, which is at most $2^{2\Delta}$.

{\bf Step 2:} With the set $\mathcal{N}_x(1)$ determined, we next enumerate all ways to place the elements in $\mathcal{N}_x(1)$ on the Hamiltonian cycle specified by $\sigma$. That is, we specify the sequence $(\sigma(n-k+1), \ldots, \sigma(n), \sigma(2), \ldots, \sigma(k+1))$, or equivalently, specify $\sigma^{-1}(j)$ for all $j\in\mathcal{N}_x(1)$. 

We start with $\mathcal{N}_x(1)\cap V(G)^c$. A vertex in $V(G)^c$ is one whose neighborhood is preserved, {\it i.e.}, $V(G)^c=\{j\in [n]: \mathcal{N}_x(j)=\mathcal{N}_{x^*}(j)\}$. For each $j\in\mathcal{N}_x(1)\cap V(G)^c$, we have by fact 3,
\[
d_x(1,j)=2k-1-\left|\mathcal{N}_x(1)\cap \mathcal{N}_x(j)\right|.
\]
Since $d_x(1,j)$ is completely determined by $\mathcal{N}_x(1)$, there are only two possibilities for $\sigma^{-1}(j)$. 

Furthermore, for every pair $j, j'\in\mathcal{N}_x(1)\cap V(G)^c$, again by fact 3,
\[
d_x(j,j')=
\begin{cases}
2k-1-\left|\mathcal{N}_x\left(j\right)\cap \mathcal{N}_x\left(j'\right)\right| & \;\;\;\text{if }j'\in \mathcal{N}_x\left(j\right);\\
2k+1-\left|\mathcal{N}_x\left(j\right)\cap \mathcal{N}_x\left(j'\right)\right| & \;\;\;\text{otherwise}.
\end{cases}
\]
So $d_x(j,j')$ is also determined by $\mathcal{N}_x(j)$ and $\mathcal{N}_x(j')$. Therefore the entire sequence $(\sigma^{-1}(j): j\in\mathcal{N}_x(1)\cap V(G)^c)$ is determined up to a global reflection around 1.

Next we handle all $j\in \mathcal{N}_x(1)\cap V(G)$. Note that $\sigma^{-1}(j)\in\{n-k+1,\ldots, n, 2,\ldots,k+1\}$ because $j\in\mathcal{N}_x(1)$. Among those $2k$ possible values, some are already taken by $\{\sigma^{-1}(j):j\in\mathcal{N}_x(1)\cap V(G)^c\}$, leaving $\left|\mathcal{N}_x(1)\cap V(G)\right|$ values to which all $j\in \mathcal{N}_x(1)\cap V(G)$ are to be assigned. The number of possible assignments is bounded by $\left|\mathcal{N}_x(1)\cap V(G)\right|!$. Since $\left|\mathcal{N}_x(1)\cap V(G)\right|\leq \min \{2k,2\Delta\}$, $\left|\mathcal{N}_x(1)\cap V(G)\right|!\leq (2k)^{2\Delta}$.

Overall, the number of possible choices of the ordered tuple $(\sigma(n-k+1), \ldots, \sigma(n), \sigma(2), \ldots, \sigma(k+1))$ is at most
\[
2\cdot2^{2\Delta}\cdot(2k)^{2\Delta}=2(4k)^{2\Delta}.
\]

{\bf Step 3:} In the previous two steps the values of $(\sigma(n-k+1), \ldots, \sigma(k+1))$ have been determined, and so are the blue edges between members of $\{\sigma(n-k+1), \ldots, \sigma(k+1)\}$. That is because $(\sigma(j),\sigma(j'))$ is a blue edge if and only if $d_{x^*}(j,j')\leq k$ and $d_{x^*}(\sigma(j),\sigma(j'))> k$. 
Denote this set of blue edges as $E_{\Blue}^{(1)}$, which can be empty. Recall that, by balancedness, the total number of blue edges in $G$ is $\Delta$. If $|E_\Blue^{(1)}|$ is already $\Delta$, then the enumeration scheme is complete because $x$ is completely specified by the difference graph. Otherwise we determine the value of $\sigma(i+k+1)$ sequentially, starting from $i=1$. At the $i$'th iteration, we first assign the value of $\sigma(i+k+1)$, the only remaining undetermined neighbor of $\sigma(i+1)$ in $x$. Then we update the set of blue edges based on the value of $\sigma(i+k+1)$: let $E_\Blue^{(i+1)}=E_\Blue^{(i)}\cup E_{\rm{update}}^{(i)}$, where
\[
E_{\rm{update}}^{(i)} \triangleq \left\{(\sigma(j), \sigma(i+k+1)): d_{x^*}(\sigma(j), \sigma(i+k+1))>k, j=i+1,..., i+k\right\}.
\]
In other words, $E_\Blue^{(i)}$ stands for the set of blue edges that have been determined after the $i-1$'th iteration. We repeat this process until all $\Delta$ blue edges are determined, {\it i.e.,} $|E_\Blue^{(i)}|=\Delta$.

At the start of the $i$'th iteration, all of $\sigma(n-k+1), \ldots, \sigma(i+k)$ have been determined. Unless $|E_{\Blue}^{(i)}|=\Delta$, specify $\sigma(i+k+1)$ as follows. 

Consider three cases according to the red degree of $\sigma(i+1)$, {\it i.e.,} the number of red edges incident to $\sigma(i+1)$ in $E_\Red$. Note that after the value of $\sigma(i+k+1)$ is assigned, $\mathcal{N}_x(\sigma(i+1))$ would be completely specified and all blue edges in $G$ that are incident to $\sigma(i+1)$ would be determined. Therefore exactly one of the following three cases must occur (for otherwise there would be more red edges than blue edges incident to $\sigma(i+1)$ in $G$, contradicting the balancedness of $G$):

\begin{enumerate}
\item (\prettyref{fig:blue.1}) The red degree of $\sigma(i+1)$ is zero, meaning that $\mathcal{N}_x(\sigma(i+1))=\mathcal{N}_{x^*}(\sigma(i+1))$. 
We claim that the value of $\sigma(i+k+1)$ has already been uniquely determined. Indeed, at the $i$th iteration, all but one members of $\mathcal{N}_x(\sigma(i+1))$ are determined, and $\sigma(i+k+1)$ has to be the true neighbor of $\sigma(i+1)$ that is not in $\{\sigma(i-k+1), ..., \sigma(i), \sigma(i+2), ..., \sigma(i+k)\}$.
\begin{figure}[H]
\centering
\includegraphics[width = 5 in]{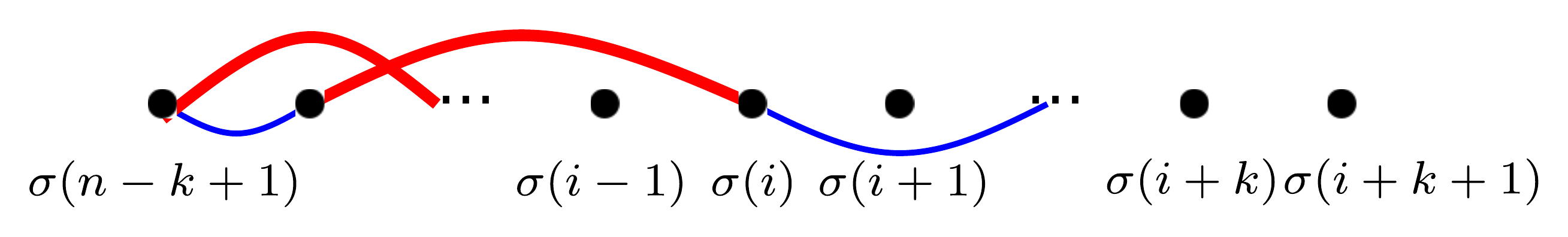}
\caption{Vertices arranged by their order on the Hamiltonian cycle corresponding to $\sigma$. At the $i$th iteration, the values of $\sigma(n-k+1)$ to $\sigma(i+k)$ are determined. The figure shows an example of case 1: the vertex $\sigma(i+1)$ is not attached to any red edges.}
\label{fig:blue.1}
\end{figure} 

\item (\prettyref{fig:blue.2}) The red degree of $\sigma(i+1)$ is nonzero and equals the number of blue edges in $E_\Blue^{(i)}$ incident to $\sigma(i+1)$. We claim that the number of possible values of $\sigma(i+k+1)$ is at most $2k$. In this case by balancedness all blue edges incident to $\sigma(i+1)$ are contained in $E_\Blue^{(i)}$ and therefore the edge $(\sigma(i+1), \sigma(i+k+1))$ does not appear in the difference graph $G$. That implies $\sigma(i+k+1)$ is connected to $\sigma(i+1)$ in $x^*$, limiting the number of choices for $\sigma(i+k+1)$ to at most $2k$.
\begin{figure}[H]
\centering
\includegraphics[width = 5 in]{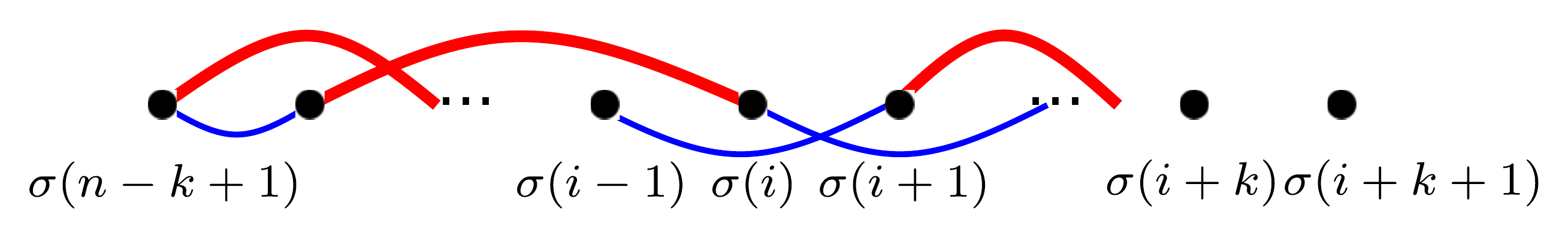}
\caption{Case 2: $\sigma(i+1)$ is attached to some red edge(s) and is already balanced at step $i$. In the figure the red degree and blue degree of $\sigma(i+1)$ are both 1, thus $(\sigma(i+1),\sigma(i+k+1))$ cannot be a blue edge in $G$. 
}
\label{fig:blue.2}
\end{figure}

\item (\prettyref{fig:blue.3}) The red degree of $\sigma(i+1)$ is nonzero and equals one plus the number of blue edges in $E_\Blue^{(i)}$ incident to $\sigma(i+1)$. By balancedness, $(\sigma(i+1), \sigma(i+k+1))$ is a blue edge in $G$. In this case, either $1\leq |E_{\rm{update}}^{(i)}|<k$ or $|E_{\rm{update}}^{(i)}|=k$. Suppose this is the $t$'th time case 3 happens. Let $\xi_t$ encode which of the two possibilities occurs and specify the value $\sigma(i+k+1)$ as follows:
\begin{enumerate}
\item Let $\xi_t=0$ and specify $\sigma(i+k+1)$ such that $1\leq |E_{\rm{update}}^{(i)}|<k$. That is, at least one of $\left\{(\sigma(j), \sigma(i+k+1)): i+2\leq j\leq i+k\right\}$ is not a blue edge in $G$. In this case $\sigma(i+k+1)$ is a true neighbor of at least one of $\{\sigma(i+2), \ldots, \sigma(i+k)\}$; in other words, $\sigma(i+k+1)\in \cup_{i+2\leq j\leq i+k}\mathcal{N}_{x^*}(j)$. Thus, the number of possibilities of $\sigma(i+k+1)$ is at most $2k(k-1)$.
\item Let $\xi_t=1$ and specify $\sigma(i+k+1)$ such that $|E_{\rm{update}}^{(i)}|=k$. That is, each one of $\left\{(\sigma(j), \sigma(i+k+1)):i+2\leq j\leq i+k\right\}$ is a blue edge in $G$. Here $\sigma(i+k+1)$ can choose from at most $|V(G)|\leq 2\Delta$ vertices.
\end{enumerate}

\begin{figure}[H]
\centering
\includegraphics[width = 5 in]{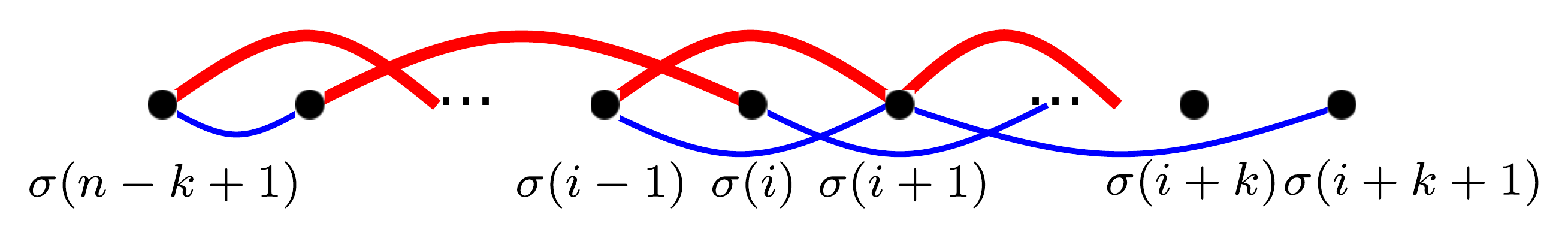}
\caption{Case 3: $\sigma(i+1)$ is attached to some red edge(s) and is not already balanced at step $i$. In the figure $\sigma(i+1)$ has red degree 2 and blue degree 1. Therefore $(\sigma(i+1),\sigma(i+k+1))$ must appear $G$ as a blue edge.}
\label{fig:blue.3}
\end{figure}
\end{enumerate}

The above process terminates until $|E_\Blue^{(i)}|=\Delta$, at which point the sequence $(\sigma(k+2),\ldots,\sigma(n-k))$ are determined. Note that each iteration, 
which one of case 1, 2 or 3 occurs is automatically determined. Therefore it suffices to enumerate (i) the value of $\sigma(i+k+1)$ at the $i$th iteration; (ii) the binary sequence $\xi$ which determines case 3a or case 3b whenever case 3 occurs.
Note that
\begin{itemize}
	\item In total, case 3b) can occur at most $\lfloor\Delta/k\rfloor$ times because $|E_\Blue^{(i)}|$ increases by $k$ each time. 
	\item Also, case 2) and case 3) combined can occur at most $2\Delta$ times, because they only occur when $\sigma(i+1)\in V(G)$.
	\item From the previous fact, the length of the $\xi$ sequence is at most $2\Delta$. 
\end{itemize}
Overall, the total number of possibilities is at most
\begin{align*}
\sum_{\xi\in \{0,1\}^{2\Delta}}(2k(k-1))^{2\Delta}(2\Delta)^{\Delta/k}
\leq \left(8k^2\right)^{2\Delta}\Delta^{\Delta/k}.
\end{align*}

Combined with the cardinality bounds from step 1 and step 2, we have
\[
\left|\mathcal{X}(E_{\Red})\right|\leq 2(4k)^{2\Delta}\cdot (8k^2)^{2\Delta}\Delta^{\Delta/k}=2(32k^3)^{2\Delta}\Delta^{\Delta/k}.
\]
\end{proof}

\section{Almost exact recovery}
\label{sec:weak}
In this section, we present our main results and proofs for almost exact recovery. 
As in \prettyref{sec:exact-suff} we let $X_i$'s and $Y_i$'s denote i.i.d.\ copies of the log-likelihood ratio $ \log \frac{\diff P_n}{\diff Q_n}$ under distribution $P_n$
and under distribution $Q_n$ respectively, with log MGFs $\psi_P(\lambda)$ and $\psi_Q(\lambda)$ defined in \prettyref{eq:psiP} and \prettyref{eq:psiQ}. Denote the Legendre transforms of the log MGFs as
\begin{align}
E_Q(\tau)  & = \psi_Q^*(\tau) \triangleq \sup_{\lambda \in \reals} \lambda \tau - \psi_Q(\lambda),  \label{eq:ratefunction}    \\
 E_P(\tau)  & =  \psi_P^*(\tau) \triangleq  \sup_{\lambda \in \reals} \lambda \tau - \psi_P(\lambda) =E_Q(\tau) -\tau.	 \nonumber
\end{align}
Then Chernoff's bound gives that for all $\tau \in [-D(Q_n\|P_n), D(P_n\|Q_n)]$ and $\ell \geq 1$, 
\begin{align}
 \prob{ \sum_{i=1}^\ell X_i \le \ell \tau } \le e^{- \ell E_P(\tau) }, \quad 
 \prob{ \sum_{i=1}^\ell Y_i \ge \ell \tau } \le  e^{- \ell E_Q(\tau)}.
\label{eq:chernoff}
\end{align}

Note that $E_P$ and $E_Q$ are convex and monotone functions, such that as $\tau$ increases from $-D(Q_n\|P_n)$ to $D(P_n\|Q_n)$, 
$E_Q(\tau)$ increases from $0$ to $D(P_n\|Q_n)$ and
$E_P(\tau)$ decreases from $D(Q_n\|P_n)$ to $0$.
The following assumption postulates a quadratic lower bound of $E_P$ at the boundary:
\begin{assumption}  \label{ass:weakest_regularity}
There exists an absolute constant $c>0$, such that for all $\eta \in [0,1]$,
\begin{align}
E_P((1-\eta) D(P_n\|Q_n)) &\geq c \eta^2 D(P_n\|Q_n).   \label{eq:divquadassump}
\end{align}
\end{assumption}
It is known that \prettyref{ass:weakest_regularity} holds whenever the distribution of $ \log \frac{\diff P_n}{\diff Q_n}$ under under $P_n$ 
 is sub-Gaussian with proxy variance $O(D(P_n\|Q_n))$ (cf.~\cite[Section III]{HajekWuXu_one_info_lim15}). In the Gaussian case where $P_n=\calN(\mu,1)$ and $Q_n=\calN(0,1)$, 
we have $E_P(\tau)= \frac{1}{8} \left( \mu+ \frac{2\tau}{\mu} \right)^2 - \tau$ and $D(P_n\|Q_n)=\mu^2/2$; thus \prettyref{eq:divquadassump} holds.

\begin{theorem}[Almost exact recovery]\label{thm:weak}
Suppose Assumption~\ref{ass:weakest_regularity} holds. 
If $k \log k = o(\log n)$ and 
\begin{align}
\liminf_{n\to \infty} \frac{k D(P_n\|Q_n) }{\log n} >1, \label{eq:weak-suff}
\end{align}
then the MLE \prettyref{eq:mle} achieves almost exact recovery. Conversely, assume that $k=O(\log n)$. If almost exact recovery is possible, then
\begin{align}
\liminf_{n\to \infty} \frac{ k D(P_n\|Q_n) }{\log n} \ge 1.\label{eq:weak-necc}
\end{align}
\end{theorem}

\prettyref{thm:weak} should be compared with the exact recovery threshold $\liminf (2\alpha_n/\log n)>1$ for $2 \leq k \leq n^{o(1)}$; the latter is always stronger, since
\[
\alpha_n=-2\log \int \sqrt{dP_ndQ_n}= -2\log \mathbb{E}_{P_n}\sqrt{\frac{dQ_n}{dP_n}}\leq -2 \mathbb{E}_{P_n}\log \sqrt{\frac{dQ_n}{dP_n}}=D(P_n\|Q_n),
\]
by Jensen's inequality. 
Unlike exact recovery, the almost exact recovery threshold is inversely proportional to $k$. Intuitively, this is because
almost exact recovery only requires one
to distinguish the latent $2k$-NN graph $x^\ast$ 
from those $2k$-NN graphs 
that differ from $x^*$ by $\Omega(kn)$ edges; in contrast, as shown in \prettyref{sec:exact.lower}, the condition for exact recovery arises from eliminating those solutions differing from $x^\ast$ by four edges.




Similar to \prettyref{thm:exact}, 
\prettyref{thm:weak} is applicable to a wide class of weight distributions, for example:
\begin{itemize}
\item for $P_n=\mathcal{N}(\mu_n,1)$ and $Q_n=\mathcal{N}(\nu_n,1)$, the sharp threshold for almost exact recovery is
\[
\liminf_{n\rightarrow \infty}\frac{k(\mu_n-\nu_n)^2}{2\log n}>1;
\]
\item for $P_n= \Pois(\mu_n)$ and $Q_n=\Pois(\nu_n)$, the sharp threshold for almost exact recovery is
\[
\liminf_{n\rightarrow\infty}\frac{k\left(\mu_n\log(\mu_n/\nu_n)+\nu_n-\mu_n\right)}{\log n}>1.
\]
\end{itemize}

The proof of \prettyref{thm:weak} follows the same strategy as that in \cite{HajekWuXu_one_info_lim15}, which studies recovering a hidden community (densely-connected subgraph) in a large weighted graph; specifically, the sufficient condition for almost exact recovery is established by analyzing the (suboptimal) MLE\footnote{For almost exact recovery, the optimal estimator that minimizes the objective $\expect{d(\hat x,x^*)}$ 
is the bit-wise maximum a posterior (MAP) estimator: $\hat{x}_e(w)=1$
if $\prob{x^*_e=1|w} \ge \prob{x^*_e=0|w}$; and $\hat{x}_e(w)=0$ otherwise.} and the necessary condition follows from a mutual information and rate-distortion argument. 
Nevertheless, as our model differs significantly from the hidden community model, the proof here requires much more sophisticated techniques, involving a delicate union bound to separate the contributions of the red edges from blue
edges and crucially relying on a sequence of counting lemmas for NN graphs shown earlier in Lemmas \ref{lmm:Xdeltacard}--\ref{lmm:blue}.

\subsection{Proof of correctness of MLE for almost exact recovery}
We abbreviate the MLE $\hat{x}_{\rm ML}$ as $\hat{x}$ in the proof below.
 For any $2k$-NN graph $x\in\calX$, recall from \prettyref{sec:exact-suff}
the difference graph $G(x)$ defined by $x-x^*$. Let $E_\Red(x)$ and $E_\Blue(x)$ denote the 
 set of red and blue edges in $G(x)$, respectively. 
 Let $\Delta= | E_\Red(\hat{x})) | = d(\hat{x},x^*)/2$. Then $ 0\le \Delta \le nk $. To prove the sufficiency, it suffices to show that 
$\prob{ \Delta \ge \epsilon_n nk} =o(1)$ for some $\epsilon_n=o(1)$ to be chosen. 

Recall that $\calX_\ell$ is the set of all $x \in \calX$ such that $G(x)$ contains exactly 
 $\ell$ red edges, i.e., $d(x, x^*)=2\ell$.  For any $ 1 \le \Delta \le nk$ and any $\tau \in \reals$, we have that
\begin{align*}
\{ \Delta=\ell \}  & \subset \{ \exists x \in \calX_\ell:  \iprod{L}{ x-x^*} >0  \} \\
& \subset \sth{   \exists x \in \calX_\ell: \sum_{e \in E_\Red(x)} L_e < \sum_{e \in E_\Blue(x)  } L_e } \\
&  \subset \sth{\exists x \in \calX_\ell: \sum_{e \in E_\Red(x)} L_e  \le \ell \tau }
\cup \sth{ \exists x \in \calX_\ell:\sum_{e \in E_\Blue(x)  } L_e \ge \ell \tau}.
\end{align*}

For each $x \in \calX_\ell,$ we have that
\begin{align*}
\sum_{e \in E_\Red(x)} L_e  \overset{d}{=} \sum_{i=1}^\ell X_i,  \quad 
\sum_{e \in E_\Blue(x)  } L_e  \overset{d}{=} \sum_{i=1}^\ell Y_i.
\end{align*}
Recall that $\mathcal{E}_\Red(\ell)=\{ E_\Red(x): x\in X_\ell\}$ stands for the set of all possible $E_\Red(x)$ where $x$ ranges over all possible $2k$-NN graphs in $\calX$
with $d(x,x^*)=2\ell$. Note that 
$$
\sth{ \exists x \in \calX_\ell:  \sum_{e \in E_\Red(x)} L_e  \le \ell \tau }
= \sth{ \exists E \in \mathcal{E}_\Red(\ell): \sum_{e \in E} L_e \le \ell \tau }. 
$$
By the union bound and Chernoff's bound \prettyref{eq:chernoff}, we get that
\begin{align}
\nonumber\prob{\Delta=\ell} & \le \left|  \mathcal{E}_\Red(\ell) \right| \prob{ \sum_{i=1}^\ell X_i \le \ell \tau }
+ \left| \calX_\ell \right|  \prob{ \sum_{i=1}^\ell Y_i \ge \ell \tau } \\
\label{eq:tilt}& \le \left|  \mathcal{E}_\Red(\ell) \right|  e^{- \ell E_P(\tau) } + \left| \calX_\ell \right|  e^{- \ell E_Q(\tau) }.
\end{align}
Note that for $\ell \ge \epsilon_n nk$, it follows from \prettyref{lmm:red} that
$$
\left|  \mathcal{E}_\Red(\ell) \right|  \le (96k^2)^\ell \binom{kn}{\lfloor \ell/2\rfloor}
\le (96k^2)^\ell \left(  \frac{ 2e nk }{\ell} \right)^{\ell/2}
\le (96k^2)^\ell \left(  \frac{ 2e  }{\epsilon_n} \right)^{\ell/2},
$$
where we used $\binom{n}{m} \le (en/m)^m$ and $\ell \le \epsilon_n nk$.
Similarly, combining \prettyref{lmm:red} and \prettyref{lmm:blue}, we have for $ \epsilon_n nk \le \ell \le nk $, 
\begin{align*}
\left| \calX_\ell \right| \le & (96k^2)^{\ell}{kn\choose \lfloor\ell/2\rfloor} \cdot 2(32k^3)^{2\ell }\ell^{\ell/k}\\
\le &  (96k^2)^{\ell} \left(\frac{2ekn}{\ell}\right)^{\ell/2}2(32k^3)^{2\ell }\ell^{\ell/k}\\
\le & (96k^2)^\ell \left(  \frac{ 2e  }{\epsilon_n} \right)^{\ell/2}  2(32k^3)^{2\ell }(nk)^{\ell/k},
\end{align*}
where the last inequality comes from the range of $\ell$. Thus for any $\epsilon_n nk \le \ell \le nk$,
$$
\prob{\Delta=\ell} 
\le e^{-\ell E_1} + e^{-\ell E_2}
$$
with
\begin{align*}
E_1 \triangleq & E_P(\tau) - \frac{1}{2} \log \frac{1}{\epsilon_n} - 2 \log k -O(1)  ,\\
E_2 \triangleq &  E_Q(\tau) - \frac{1}{k} \log n -  \frac{1}{2} \log \frac{1}{\epsilon_n}  - 8 \log k - O(1).
\end{align*}
By \prettyref{eq:weak-suff}, we have $k D(P_n\|Q_n) (1-\eta) \geq \log n $ for some $\eta \in (0,1)$. 
Choose $\tau = (1-\eta) D(P_n\|Q_n)$.
By the assumption \prettyref{eq:divquadassump}, we have
\begin{align*}
E_1 \geq c  \eta^2 D(P_n\|Q_n)   - \frac{1}{2} \log \frac{1}{\epsilon_n} - 2 \log k -O(1).
\end{align*}
Using the fact that $E_P(\tau)=E_Q(\tau)-\tau$, we have
\begin{align*}
E_2 & \geq c \eta^2   D(P_n\|Q_n)   - \frac{1}{2}  \log \frac{1}{\epsilon_n}
 +  D(P_n\|Q_n) (1-\eta) - \frac{1}{k} \log n - 8 \log k - O(1) \ \\
 & \geq c \eta^2  D(P_n\|Q_n)    - \frac{1}{2}  \log \frac{1}{\epsilon_n}  - 8 \log k - O(1).
\end{align*}
Since $k \log k =o(\log n)$ and $k D(P_n\|Q_n) \geq \log n $, it follows that $ D(P_n\|Q_n) = \omega(  \log k)$. 
Therefore, setting $\epsilon_n = 1/ \left( k D( P_n \| Q_n) \right) $, it follows that 
$E \triangleq \min\{E_1,E_2\} = \Omega(D(P_n\|Q_n))$. 
Hence, 
\begin{align*}
\prob{\Delta \ge \epsilon_n nk }  & = \sum_{\ell \ge \epsilon_n nk }^{nk} \prob{\Delta = \ell}  \\
& \leq \sum_{\ell=\epsilon_n nk }^{\infty} \left( \eexp^{-\ell E_1} + \eexp^{-\ell E_2} \right)  \\
&  \leq \frac{2 \exp(-\epsilon_n n k  E)}{1 - \exp(-E)} = \exp\left( - \Omega(n) \right).
\end{align*}
In other words, the MLE achieves almost exact recovery.


\begin{remark}\label{rmk:measures}
This is a good place to explain how the two distance measures between $P_n$ and $Q_n$ show up for the conditions for exact recovery and almost exact recovery, respectively. The two types of recovery demand control of $\mathbb{P}\{\Delta=\ell\}$ for different ranges of $\ell$. For almost exact recovery, we need to control 
$\mathbb{P}\{\Delta=\ell\}$ for
$\ell\geq \epsilon nk$. In this range, there is a large difference between $|\mathcal{E}_\Red(\ell)|$ and $|\mathcal{X}_\ell|$. Indeed from \prettyref{lmm:blue}, there may be up to $(ck)^{2\ell}\ell^{\ell/k}$ members of $\mathcal{X}_\ell$ with the same set of red edges. Hence for large $\ell$, it is more advantageous to separate the contributions from the red edges and blue edges, as done in~\eqref{eq:tilt}. To balance out the two terms in~\eqref{eq:tilt}, the exponential tilting parameter $\tau$ is chosen so that $E_Q(\tau)$ is large. Given that $E_Q(\tau)$ is an increasing function on $[-D(Q_n\| P_n), D(P_n\| Q_n)]$, we choose $\tau$ close to $D(P_n\| Q_n)$ with $E_Q(\tau)\approx D(P_n\| Q_n)$. As a result, the condition for almost exact recovery emerges from a battle between $D(P_n\| Q_n)$ and $|\mathcal{X}_\ell|$.

Exact recovery, on the other hand, calls for upper bounds on $\mathbb{P}\{\Delta=\ell\}$ for all $\ell\geq 2$. In fact as seen from the lower bound proof of Theorem~\ref{thm:exact}, the bottleneck for exact recovery happens at $\ell=2$, where $|\mathcal{E}_\Red(\ell)|$ and $|\mathcal{X}_\ell|$ are around the same order. In this regime there is no longer any gain in separating the red edge and blue edge sets, and it is more favorable to directly invoke the following 
Chernoff bound:
\begin{align*}
\mathbb{P}\{\Delta=\ell\}\leq & \left|\mathcal{X}_\ell\right|\mathbb{P}\left\{\sum_{i\leq \ell}X_i\leq \sum_{i\leq \ell}Y_i\right\}\\
\leq & \left|\mathcal{X}_\ell\right|\mathbb{P}\left\{-\ell \inf_{\tau}(E_P(\tau)+E_Q(\tau))\right\} = \left|\mathcal{X}_\ell\right| e^{-\ell\alpha_n}.
\end{align*}
See~\cite[page~29]{bagaria2018hidden} for a derivation of the equality $\inf_{\tau}(E_P(\tau)+E_Q(\tau))=\alpha_n$. As a result, the condition for exact recovery is governed by the distance $\alpha_n$.
\end{remark}

\subsection{Information-theoretic lower bound for almost exact recovery}
Suppose that almost exact recovery of $x^*$ is achieved by some estimator $\hat{x}$, such that $\expect{d(\hat{x},x^*)}=2nk\epsilon_n$, for some $\epsilon_n\to 0$. 
We show that \prettyref{eq:weak-necc} must hold.
First, we can assume, WLOG, $\hat{x}$ takes value in $\mathcal{X}$, the set of all $2k$-NN graphs. Indeed, if we set 
\begin{equation}
\hat{x}'=\argmin_{x\in \mathcal{X}} d(x,\hat{x}),
\label{eq:hatx-mod}
\end{equation}
 then 
$d(\hat{x}',x^*)\leq d(\hat{x}',\hat{x})+ d(\hat{x},x^*) \leq 2 d(\hat{x},x^*)$ and hence 
$\expect{d(\hat{x}',x^*)}\leq 4nk\epsilon_n$;
in other words, $\hat{x}'$ also achieves almost exact recovery. 

Since $x^*\to w \to \hat x$ form a Markov chain, by the data processing inequality of mutual information, we have
\begin{align}
I(w; x ^*)\ge I(\hat{x}, x ^*)  & \ge \inf\{I(\tilde{x},x^*): \tilde x \in \calX, \expect{d(\tilde{ x } ,x^*)} \le 2nk  \epsilon_n  \} \label{eq:dpi1} \\
&= H(x^*) - \sup\{H(x^* | \tilde{x} ): \tilde x \in \calX, \expect{d(\tilde{ x } ,x^*)} \le 2nk  \epsilon_n  \} \label{eq:dpi2} 
\end{align}
where the infimum in \prettyref{eq:dpi1}, known as the rate-distortion function, is taken over all conditional distributions $P_{\tilde x|x^*}$ satisfying the constraints.
Note that $H(x^*)=\log (n!)= (1+o(1)) n \log n$. Moreover from~\prettyref{lmm:red} and~\prettyref{lmm:blue},
for any fixed $\tilde{x} \in \calX$, the number of possible  $x^* \in \calX$ such that 
$d(\tilde{x}, x^*)=2\ell$ is at most 
$$
 (96k^2)^{\ell}{kn\choose \lfloor\ell/2\rfloor}\cdot 2(32k^3)^{2\ell }\ell^{\ell/k}
\le 2 (c_k n/\ell )^{\ell/2} (nk) ^{ \ell/k} ,
$$
where $c_k=2e 96^2 2^{20} k^{17}$ and we used the fact that $\ell \le nk$. 
Therefore,
$$
H( x^* | \tilde{x }, d(\tilde{x },x ^*)=2\ell )  \le \frac{\ell}{2} \log ( c_k n)  -  \frac{\ell}{2} \log \ell + \frac{\ell}{k} \log (nk) +  \log 2.
$$
By the convexity of $x\mapsto x \log x$ and Jensen's inequality, it follows that 
$$
H(x^* | \tilde{x}, d(\tilde{x}, x^*)) \le  \frac{\expect{ d(\tilde{ x } ,x^*)} }{4} \sth{ \log ( c_k n) 
- \log \frac{ \expect{ d(\tilde{ x } ,x^*)}}{2} + \frac{2  }{k} \log (nk)}  +  \log 2.
$$
Furthermore, $d(x,x^*)$ takes values in $\{0,\ldots,2nk\}$. Thus from the chain rule, 
\[
H(x^* | \tilde{x } ) = H(d(\tilde{x}, x^*)| \tilde{x})+H(x^* | \tilde{x}, d(\tilde{x}, x^*)) \leq 
\log(1+2nk) + H(x^* | \tilde{x}, d(\tilde{x}, x^*)),
\] 
and hence
\begin{align*}
& \sup\{H(x^* | \tilde{x} ): \tilde x \in \calX, \expect{d(\tilde{ x } ,x^*)} \le 2nk  \epsilon_n  \} \\
\le & \frac{ 1}{2} \epsilon_n nk  \left[ \log ( c_k n ) - \log \left( \epsilon_n nk  \right) \right] +  \epsilon_n n   \log ( nk)  
 +  \log 2 + \log (1+2nk) = o(n \log n),
\end{align*}
where the last equality holds due to the assumption that $k\log k=o(\log n)$. 
Therefore, we get from \prettyref{eq:dpi2} 
that $I(w;x^* ) \ge (1+o(1)) n \log n$. On the other hand,
$$
I(w; x^*) = \min_{Q_w} \Expect_{x^*}[D( {P}_{w|x^*} \| Q_w)]
\le \Expect_{x^*}[D( {P}_{w|x^*}  \| Q_n^{\otimes \binom{n}{2} })]
 = nk D(P_n\| Q_n ),
$$
where the minimum is taken over all distribution $Q_w$, achieved at $Q_w=P_w$.
This yields the desired $k D(P_n\| Q_n ) \geq (1+o(1))\log n$.

\section{Discussion on efficient recovery algorithms}
\label{sec:discussion}

As shown in~\prettyref{sec:exact} and~\prettyref{sec:weak}, the sharp thresholds for exact and almost exact recovery can both be attained by the MLE \prettyref{eq:mle}, which, however, entails solving the computationally 
intractable max-weight $2k$-NN subgraph problem. 
So far no polynomial-time algorithm is known to achieve the sharp thresholds for exact or almost exact recovery except when $k=1$~\cite{bagaria2018hidden}. 
In Section~\ref{sec:alg.general}, we consider several computationally efficient algorithms to recover the hidden $2k$-NN graph and analyze their statistical properties.
In Section~\ref{sec:alg.bernoulli}, we focus on the special case of small-world graphs where the edge weights are distributed Bernoulli and give a polynomial time algorithm that achieves the sharp threshold for exact recovery.

\subsection{Efficient recovery algorithms under the general hidden $2k$-NN graph model}\label{sec:alg.general}
For simplicity we focus on the Gaussian model with weight distribution $P_n=\mathcal{N}(\mu_n,1)$ and $Q_n=\mathcal{N}(0,1)$ for $\mu_n>0$. Analysis in this section can be extended to general weight distributions.
From Theorems \ref{thm:exact} and~\ref{thm:weak}, under the Gaussian model, the sharp thresholds for exact recovery
(for $2 \le k \le n^{o(1)}$)
 and almost exact recovery (for $1 \le k \le o(\log n/\log \log n)$) are 
\begin{equation}
\liminf_{n\rightarrow\infty}\frac{\mu_n^2}{2\log n}>1
\label{eq:exact-gaussian}
\end{equation}
and
\begin{equation}
\liminf_{n\rightarrow\infty}\frac{k\mu_n^2}{2\log n}>1,
\label{eq:almostexact-gaussian}
\end{equation}
respectively.
Since the log likelihood ratio is given by $L_e=\log \frac{\diff P_n}{\diff Q_n}(w_e)=\mu_nw_e-\mu_n^2$, the MLE \prettyref{eq:mle} can be equivalently written as
\begin{align}
\widehat{x}_{\text{ML}} = \argmax_{x \in \calX} \langle w,x\rangle.
\label{eq:mle_gaussian}
\end{align}
When $k=1$, the formulation~\prettyref{eq:mle_gaussian} is known as \emph{max-weighted Hamiltonian cycle};
the previous work~\cite{bagaria2018hidden} analyzes its $2$-factor integer linear program (ILP) relaxation  and fractional $2$-factor linear program (LP) relaxation, and show that they achieve
the sharp exact recovery threshold $\liminf_{n\rightarrow\infty}\frac{\mu_n^2}{4\log n}>1$.
This motivates us to consider the ILP relaxation and LP relaxation for general $k$.

\paragraph*{$2k$-factor ILP relaxation} 
By relaxing the $2k$-NN graph constraint in the MLE~\prettyref{eq:mle_gaussian} 
to a degree constraint, 
we arrive at the following $2k$-factor ILP:
\begin{align}
\widehat{x}_{2k\text{F}}= \argmax_x & \; \langle w,x\rangle \label{eq:2kF}\\
\text{s.t.   } &  \sum_{v\sim u}x_{(u,v)}=2k, \;\; \forall u,\nonumber\\
& x_e\in \{0,1\},\;\; \forall e\nonumber
\end{align}
where the first constraint enforces that every vertex has degree $2k$. 
It is known that for constant $k$, the 
ILP \prettyref{eq:2kF} can be solvable in $O(n^4)$ time~\cite{letchford2008odd}.
 
To analyze $\widehat{x}_{2k\text{F}}$, note that each feasible solution $x$ to the ILP is a $2k$-regular graph. Therefore, the difference graph $x-x^*$ is still balanced and the simple bound~\eqref{eq:Xdeltacard_simple} continues hold:
\[
\left|\mathcal{X}_\Delta \right|\leq (4kn)^\Delta,
\]
where $\mathcal{X}_\Delta$ is the collection of $2k$-regular graphs $x$ such that the difference graph $x-x^*$ contains exactly $\Delta$ red edges.
Moreover, for $x\in \mathcal{X}_\Delta$, $\langle w,x-x^*\rangle\sim \mathcal{N}(-\Delta\mu_n, 2\Delta)$. Hence from the union bound
\[
\mathbb{P}\{\widehat{x}_{2k\text{F}}\neq x^*\}\leq \sum_{\Delta\geq 1}(4kn)^\Delta \exp\left(-\frac{\Delta\mu_n^2}{4}\right)=\sum_{\Delta\geq 1}\exp\left(-\Delta\left(\frac{\mu_n^2}{4}-\log (4kn)\right)\right),
\]
we conclude that when $2\leq k\leq n^{o(1)}$, $\widehat{x}_{2k\text{F}}$ achieves exact recovery if $\liminf_{n\rightarrow\infty}\mu_n^2/(4\log n)>1$, which is suboptimal by a multiplicative
factor of $2$ compared to the sharp threshold \prettyref{eq:exact-gaussian}.

In fact, $\widehat{x}_{2k\text{F}}$ provably fails to attain the sharp exact recovery threshold when $2\leq k\leq n^{o(1)}$.
To see this, assume that $x^*$ is associated with the identity permutation and consider its modifications of the following form: fix two vertices $i,j$ for which $d_{x^*}(i,j)>k$, remove the edges $(i,i+1)$ and $(j,j+1)$ in $x^*$ and add the edges $(i,j)$ and $(i+1,j+1)$, resulting in a $2k$-regular graph $x^{(i,j)}$ feasible to \prettyref{eq:2kF}. 
There are $O(n^2)$ such modified solutions and the difference in weights $\Iprod{w}{x^{(i,j)}-x^*}$
are close to being mutually independent. Each $x^{(i,j)}$ corresponds to a difference graph with $\Delta=2$ red edges. Hence $\langle w,x^{(i,j)}-x^*\rangle\sim \mathcal{N}(-2\mu_n, 4)$. 
By following the similar lower bound proof for exact recovery in~\prettyref{sec:exact.lower},  
we can conclude that if  $\liminf_{n\rightarrow\infty}\mu_n^2/(4\log n)<1$, 
then with high probability there is at least one feasible solution $x^{(i,j)}$ such that 
$\langle w,x^{(i,j)}-x^* \rangle>0$, yielding $\widehat{x}_{2k\text{F}} \neq x^*$.

\paragraph*{LP relaxation} By further relaxing the integer constraint in $\widehat{x}_{2k\text{F}}$, we arrive at the following LP:
\begin{align*}
\widehat{x}_{\text{LP}}= \argmax_x & \; \langle w,x\rangle\\
\text{s.t.   } &  \sum_{v\sim u}x_{(u,v)}=2k, \;\; \forall u,\\
& x_e\in [0,1],\;\; \forall e.
\end{align*}
Since $\widehat{x}_{\text{LP}}$ is a relaxation of $\widehat{x}_{2k\text{F}}$,
it follows from the above negative result of ILP that  
$\widehat{x}_{\text{LP}} \neq x^*$ when $\liminf_{n\rightarrow\infty}\mu_n^2/(4\log n)<1$. 
In the positive direction, one can show that $\widehat{x}_{\text{LP}}$ 
also achieves exact recovery for $1\leq k\leq n^{o(1)}$ 
when $\liminf_{n\rightarrow\infty}\mu_n^2/(4\log n)>1$.
That is because firstly, solutions to the LP must be half-integral ($x_e\in \{0,1/2,1\}$ for all $e$). Moreover, the difference graph $x-x^*$ can be represented by a balanced multigraph with edge multiplicity at most 2 (we refer the reader to~\cite{bagaria2018hidden} for details). The rest of the proof follows exactly from the proof of~\cite[Theorem~1]{bagaria2018hidden}.

To sum up, when it comes to exact recovery, the statistical performance for $\widehat{x}_{2k\text{F}}$ and $\widehat{x}_{\text{LP}}$ match in the asymptotics. They both require $\liminf_{n\rightarrow\infty}\mu_n^2/(4\log n)>1$, which is suboptimal by a factor of two. 
Whether they can achieve almost exact recovery under weaker conditions remains open.

\paragraph*{Simple thresholding} 
To partially address the problem of almost exact recovery, we consider a na\"ive thresholding estimator $\widehat{x}_{\text{TH}}$ given by
\[
\widehat{x}_{\text{TH}}(e)=\mathds{1}\left\{w_e>\sqrt{(2+\epsilon_n)\log n}\right\},
\]
where $\epsilon_n$ is any sequence such that $\epsilon_n=o(1)$ and $\epsilon_n = \omega(\frac{1}{\log n})$.
For each edge $e$ in the true $2k$-NN graph $x^*$, $w_e\sim \mathcal{N}(\mu_n,1)$ and thus
\[
\mathbb{P}\{\widehat{x}_{\text{TH}}(e)=0\}\leq \exp\left(-(\mu_n-\sqrt{(2+\epsilon_n)\log n})^2/2\right);
\]
Similarly for edge $e$ not in $x^*$,
\[
\mathbb{P}\{\widehat{x}_{\text{TH}}(e)=1\}\leq \exp(-(\sqrt{(2+\epsilon_n)\log n})^2/2)= n^{-(1+\epsilon_n/2)}.
\]
Recall that $d(\widehat{x}_{\text{TH}}, x^*)=\sum_e \mathds{1}\{\widehat{x}_{\text{TH}}(e)\neq x^*(e)\}$. We have
\[
\mathbb{E}\left[ d \left(\widehat{x}_{\text{TH}}, x^* \right) \right]
\leq kn\exp\left(-(\mu_n-\sqrt{(2+\epsilon_n)\log n})^2/2\right) + n^2\cdot n^{-(1+\epsilon_n/2)}.
\]
Since $\epsilon_n\rightarrow 0$, the first term is of order $o(nk)$ when $\liminf_{n\rightarrow\infty} \mu_n^2/(2\log n)>1$.
The second term is of order $o(n)$ as $\epsilon_n \log n \to +\infty$.
In other words, the thresholding estimator  $\widehat{x}_{\text{TH}}$ achieves almost exact recovery under the assumption $\liminf  \mu_n^2/(2\log n)>1$, which is optimal for $k=1$ in view of \prettyref{eq:almostexact-gaussian}. 

It is worth pointing out that $\widehat{x}_{\text{TH}}$ may not be a valid $2k$-NN graph. One can of course consider the modified estimator \prettyref{eq:hatx-mod} by projecting $\hat x$ to the set of $2k$-NN graphs; however, it is unclear whether this can be done in polynomial time. It is an interesting open problem whether a computationally efficient 
$2k$-NN graph estimator can be obtained from $\widehat{x}_{\text{TH}}$ and still inherits the almost exact recovery guarantee $\liminf_{n\rightarrow\infty} \mu_n^2/(2\log n)>1$.

\paragraph*{Spectral methods}
For a variety of problems such as clustering and community detection, spectral methods have been successfully used to recover the hidden structures
based on the principal eigenvectors of the observed graph.
In our model, for the weighted adjacency matrix $W$,  the principal eigenvectors of $\expect{W}$ contain perfect information about the hidden
$2k$-NN graph. Indeed,  to see this, in the Gaussian model the weighted adjacency matrix $W$ can be expressed as 
$$
W= \mu_n  \Pi B \Pi^\top + Z,
$$
where $\Pi$ is the permutation matrix associated with the
hidden $2k$-NN graph;  $B$ is the adjacency matrix of the basic $2k$-NN graph where 
$B_{ij}=1$ if $\min \{|i-j|, n- |i-j|\}\leq k$ and $B_{ij}=0$ otherwise; and $Z$ is a symmetric Gaussian
matrix with zero diagonal and $Z_{ij}=Z_{ji}$ independently drawn from $\calN(0,1)$ for $i<j$.
Since $B$ is a circulant matrix, its eigenvalues and eigenvectors 
can be determined by the discrete Fourier transform of the window function:
\[
\lambda_j = \sum_{\ell=-k}^k \exp\left(\frac{2\pi i j \ell}{n}\right) -1 =
\frac{\sin \frac{(2k+1)j\pi}{n}}{\sin \frac{j\pi}{n}} - 1, \quad j=0,\ldots, n-1
\]
where $i=\sqrt{-1}$ is the imaginary unit, $\lambda_0=2k$ (degree) and $\lambda_{n-j}=\lambda_j$ which decays similarly to the $\text{sinc}$ function. 
Furthermore, the eigenvector $v_1$ of $\Pi B \Pi^\top$ corresponding to $\lambda_1$ encodes the permutation matrix $\Pi$ perfectly
as
$
v_1 = \Pi (\omega^{0}, \ldots,\omega^{n-1})^\top,
$
where $\omega = \exp \left( i \frac{2 \pi  } {n} \right)$ is the $n^{\rm th}$ root of unity.
Thus in the noiseless case one can exactly recover the underlying permutation $\Pi$ and hence the hidden $2k$-NN graph from $v_1$.
Unfortunately,  it turns out that $v_1$ is very sensitive to the noise perturbation due to the small eigengap.
In particular, the eigengap of $\lambda_1-\lambda_2 \sim k^3/n^2$, while the 
spectral norm of the noise perturbation $\|Z\|_2$ is on the order of $\sqrt{n}$. 
Thus, by the Davis-Kahan theorem, the second eigenvector of $W$ is close to $v_1$ if
$\mu_n k^3 /n^2 \gtrsim \sqrt{n}$, \ie, $\mu_n \gtrsim n^{5/2}/k^3$.\footnote{In fact, following the analysis of spectral ordering in~\cite{cai2017detection}, one can show that the almost exact recovery can be efficiently achieved via sorting the entry-wise angles of the second eigenvector of $W$ under a slightly higher SNR: $\mu_n \gg n^{7/2}/k^4$.} 
This is highly suboptimal when $k=n^{o(1)}$. 

\subsection{Achieving the exact recovery threshold under the small-world model}\label{sec:alg.bernoulli}
Although designing efficient algorithms that achieve the sharp thresholds appears challenging under the general hidden $2k$-NN graph model, the task turns out to be more manageable for the special case of the Watts-Strogatz small-world graph model where the edges are unweighted. Recall the special case \prettyref{eq:smallworld} considered in the introduction with $P_n=\Bern(p)$ and $Q_n=\Bern(q)$, where
\begin{equation}
p=1-\epsilon + \frac{2\epsilon k}{n-1} \quad \text{and}  \quad q= \frac{2\epsilon k}{n-1}.
\label{eq:smallworld1}
\end{equation}
The observed graph $w\in\{0,1\}^{{n\choose 2}}$ can be viewed as a noisy version of the true $2k$-NN graph $x^*$. By~\prettyref{thm:exact}, for $2\leq k\leq n^{o(1)}$, the sharp threshold for exact recovery is $\lim\inf_{n\rightarrow\infty} (-2\log \epsilon/\log n)>1$, i.e., $\epsilon \leq n^{-\frac{1}{2} - \Omega(1)}$.
 We show that the greedy algorithm below succeeds under this condition. 
As mentioned in \prettyref{sec:intro}, to exactly recover
$x^*$, it suffices to recover the corresponding Hamiltonian cycle identified by a permutation $\sigma^*$. Similar to the enumeration scheme that lies at the heart of the proof of~\prettyref{lmm:blue}, the algorithm first determines the neighborhood of one vertex and their ordering on the Hamiltonian cycle, and then sequentially finds the remaining vertices to complete the cycle.

\begin{algorithm}[H]
Start from an arbitrary vertex $i_0$ and let $\hat{\sigma}(1)=i_0$\;
{\bf Step 1 (label the neighbors of $\hat{\sigma}(1)$):}\\
Let $\mathcal{N}\triangleq \mathcal{N}_w(\hat{\sigma}(1))$ be the set of vertices incident to $\hat{\sigma}(1)$ in $w$\;
\eIf{The subgraph of $w$ induced by $\mathcal{N}$ is isomorphic to the subgraph of $x^*$ induced by $\mathcal{N}_{x^*}(1)$}
{Use the subgraph of $w$ induced by $\mathcal{N}$ to determine (up to a reversal) the ordering $(\hat{\sigma}(n-k+1), ...,\hat{\sigma}(n),\hat{\sigma}(2),..., \hat{\sigma}(k+1))$}
{Report error and terminate.}
{\bf Step 2 (label the remaining vertices sequentially):}\\
$\mathcal{V}_{labeled}\triangleq  \{\hat{\sigma}(n-k+1), ..., \hat{\sigma}(n), \hat{\sigma}(1), ..., \hat{\sigma}(k+1)\}$\;
\For {i=1 to n-2k-1} {
$\mathcal{U}\triangleq \mathcal{N}_w(\hat{\sigma}(i+1))\backslash \mathcal{V}_{labeled}$\;
	\Switch{$|\mathcal{U}|$}{
            	\Case{$|\mathcal{U}|\geq 2$}{
			\eIf{exactly one member $u$ of $\mathcal{U}$ is incident to $\hat{\sigma}(i+2)$} {Set $\hat{\sigma}(i+k+1)=u$.} {Report error and terminate.}
            	}
            	\Case{$|\mathcal{U}|=1$}{
            		Set $\hat{\sigma}(i+k+1)$ be the vertex in $\mathcal{U}$\;
            	}
            	\Case{$|\mathcal{U}|=0$}{
			\eIf{exactly one member $v$ of $\mathcal{N}_w(\hat{\sigma}(i+2))$ is incident to exactly $k$ members of $\mathcal{V}_{labeled}^c$}
			{Set $\hat{\sigma}(i+k+1)=v$}
			{Report error and terminate.}
            	}
        }
        $\mathcal{V}_{labeled}\triangleq  \mathcal{V}_{labeled}\cup \{\hat{\sigma}(i+k+1)\}$\;
}
Output $x(\hat{\sigma})$, the $2k$-NN graph corresponding to $\hat{\sigma}$.
\caption{Greedy algorithm for exact recovery under the small-world model}
\label{alg:exact}
\end{algorithm}

Since it suffices to recover $\sigma^*$ up to cyclic shifts and reversals,
we can assume WLOG that $i_0=\sigma^*(1)$. In Step 1 of~\prettyref{alg:exact}, one needs to order the members of $\calN$ from the subgraph of $x^*$ induced by $\mathcal{N}$.
We will show, with high probability, $\mathcal{N}=\calN_w(i_0)$ coincides with the true neighborhood $\mathcal{N}_{x^*}(\hat{\sigma}(1))$ and the subgraph of $w$ induced by $\mathcal{N}$ 
is identical to that of $x^*$. Therefore, we can infer the ordering of members of $\mathcal{N}$
using the nearest-neighbor structure of $x^*$. In particular, $\sigma^*(n-k+1)$ and $\sigma^*(k+1)$ are the only two vertices in $\mathcal{N}$ that are incident to exactly $k-1$ vertices in $\mathcal{N}$; having determined $\sigma^*(n-k+1)$, $\sigma^*(n-k+2)$ is the only vertex in $\mathcal{N}$ that is incident to $\sigma^*(n-k+1)$ and exactly $k-1$ other vertices in $N$; similarly $\sigma^*(n-k+3)$ can be uniquely determined given $\sigma^*(n-k+1)$ and $\sigma^*(n-k+2)$, so on and so forth.

Step 1 of~\prettyref{alg:exact} relies on the fact that with high probability, $w$ and $x^*$ completely agree in the neighborhood near the fixed vertex $i_0$. This does not hold uniformly for all vertices. However we will show that uniformly for all $i\in [n]$, $w$ and $x^*$ differ by at most one edge in the neighborhood near $i$. This fact is crucial for the success of the second step. Now we present the exact recovery guarantee of the algorithm.

\begin{theorem}\label{thm:exact.alg}
Assume that $k=n^{o(1)}$ and $\lim\inf_{n\rightarrow\infty} (-2\log \epsilon/\log n)>1$, then with probability $1-o(1)$, \prettyref{alg:exact} runs successfully
and returns $x(\hat{\sigma})=x^*$. 
In other words, \prettyref{alg:exact} achieves exact recovery.
\end{theorem}
\begin{proof}
Let us start with some notation. Recall that the set of true neighbors of a vertex $j$ is denoted as $\mathcal{N}_{x^*}(j)$. Let $E_{x^*}(j)$ (resp.~$E_w(j)$) denote the set of edges in $x^*$ (resp.~$w$) whose endpoints contain at least one member of $\mathcal{N}_{x^*}(\sigma^*(j))\cup \{\sigma^*(j)\}$. 
We claim that under the assumption $\lim\inf_{n\rightarrow\infty} (-2\log \epsilon/\log n)>1$, the following events occur simultaneously with high probability:
\begin{itemize}
\item $\mathcal{A}_n = \{E_w(1) = E_{x^*}(1)\}$; 
\item $\mathcal{B}_n = \{|E_w(j)\Delta E_{x^*}(j)|\leq 1,\forall j\}$;
\item $\mathcal{C}_n = \{w(\sigma^*(j+1), \sigma^*(j+k+2))=0,\forall j = 1,..., n-2k-1\}$.
\end{itemize}
First we argue that on $\mathcal{A}_n\cap \mathcal{B}_n\cap \mathcal{C}_n$,~\prettyref{alg:exact} correctly recovers $x^*$. Under $\mathcal{A}_n$, the subgraphs of $w$ induced by $\mathcal{N}_{x^*}(\sigma^*(1))$ coincides with that of $x^*$. 
Hence~\prettyref{alg:exact} successfully recovers $(\sigma^*(n-k+1),..., \sigma^*(n),\sigma^*(2),...,\sigma^*(k+1))$ up to a reversal. WLOG say $\hat{\sigma}(i)=\sigma^*(i)$ for all $\sigma^*(i)\in \mathcal{N}_{x^*}(\sigma^*(1))$. 

Next we show inductively that the algorithm correctly labels all the remaining vertices. Start from the inductive hypothesis that $\hat{\sigma}(j)=\sigma^*(j)$ for all $j\leq i+k$. Recall that $\mathcal{U}=\mathcal{N}_w(\hat{\sigma}(i+1))\backslash \mathcal{V}_{labeled}$ and the algorithm considers the following three cases according to the size of $\mathcal{U}$:
\begin{enumerate}
\item $|\mathcal{U}|\geq 2$. Under $\mathcal{B}_n$, we must have $|\mathcal{U}|=2$ because otherwise $E_w(i+1)\backslash E_{x^*}(i+1)$
contains at least two edges, contradicting $\calB_n$. Write $\mathcal{U}=\{u,v\}$. One of $u,v$ must be $\sigma^*(i+k+1)$, because otherwise $E_w(i+1)\backslash E_{x^*}(i+1)$ contains at least the two edges $(\sigma^*(i+1),u)$ and $(\sigma^*(i+1),v)$, contradicting $\mathcal{B}_n$. Say $u=\sigma^*(i+k+1)$. Then $(\sigma^*(i+1),u)$ is the only member of $E_w(i+1)\Delta E_{x^*}(i+1)$.
Thus $u$ and $\sigma^*(i+2)$ must be neighbors in $w$.
Under $\mathcal{C}_n$, $v$ cannot be $\sigma^*(i+k+2)$. Hence $v$ and $\sigma^*(i+2)$ are not neighbors in $x^*$.
Consequently, they cannot be neighbors in $w$, because otherwise both $(\sigma^*(i+2),v)$ and $(\sigma^*(i+1),v)$ belong
to $E_w(v)\Delta E_{x^*}(v)$, contradicting $\calB_n$. Under the induction hypothesis, 
$\hat{\sigma}(i+2)=\sigma^*(i+2)$. Hence, 
$u$ is the only member of $\mathcal{U}$ that is incident to $\hat{\sigma}(i+2)$ in $w$
and thus the algorithm successfully identifies $u$ as $\sigma^*(i+k+1)$.  
\item $|\mathcal{U}|=1$. In this case the element $u$ in $\mathcal{U}$ must be $\sigma^*(i+k+1)$ because otherwise both $(\sigma^*(i+1), \sigma^*(i+k+1))$ and $(\sigma^*(i+1), u)$ are contained in $E_w(i+1)\Delta E_{x^*}(i+1)$, contradicting $\mathcal{B}_n$.
\item $|\mathcal{U}|=0$. Under $\mathcal{B}_n$, $(\sigma^*(i+1), \sigma^*(i+k+1))$ is the only member of $E_w(i+1)\Delta E_{x^*}(i+1)$. As a result, $\sigma^*(i+2)$ must be incident to all of its $2k$ true neighbors. These $2k$ neighbors contain $\sigma^*(i+k+1)$, and under $\mathcal{B}_n$, $\sigma^*(i+k+1)$ is the only one that is incident to exactly $k$ members of $\mathcal{V}_{labeled}^c$. Thus the algorithm can always identify $\sigma^*(i+k+1)$.
\end{enumerate}

It remains to show that all the following three events occur with high probability. 
Under the assumption $\lim\inf_{n\rightarrow\infty} (-2\log \epsilon/\log n)>1$, there exists some positive constant $\eta$ such that $\epsilon<n^{-1/2-\eta}$ for large enough $n$. By \prettyref{eq:smallworld1}, $1-p<n^{-1/2-\eta}$ and $q<n^{-3/2-\eta+o(1)}$ for $k\leq n^{o(1)}$. 

\begin{itemize}
	
\item The event $\mathcal{A}_n$: To show $\prob{\mathcal{A}_n^c}=o(1)$, note that there are $O(k^2)$ edges in $E_{x^*}(1)$. The probability that one of them does not appear in $w$ is upper bounded by $O(k^2)\cdot (1-p)\leq n^{-1/2-\eta+o(1)}=o(1)$. Similarly the probability that an false edge shows up in $E_w(1)$ is at most $O(nk)\cdot q\leq n^{-1/2-\eta+o(1)}=o(1)$. Thus $E_{x^*}(1)=E_w(1)$ with high probability. 

\item The event $\mathcal{B}_n$:
Similar as above, $|E_{x^*}(j)\Delta E_w(j)|$ equals in distribution to $X+Y$ with $X\sim \text{Binom}(n_1, 1-p)$, $Y\sim \text{Binom}(n_2, q)$ with $n_1=O(k^2)$, $n_2=O(nk)$ and $X,Y$ independent. Thus
\[
\mathbb{P}\left\{|E_{x^*}(j)\Delta E_w(j)|>1\right\} = \mathbb{P}\{X+Y>1\}
\leq \mathbb{P}\{X>1\}+\mathbb{P}\{Y>1\} + \mathbb{P}\{X=Y=1\}.
\]
Using the Binomial distributions of $X,Y$, the above can be further bounded by 
\[
{n_1\choose 2}(1-p)^2+{n_2\choose 2}q^2+n_1n_2(1-p)q=o(1/n).
\]
By the union bound,
\[
\prob{\mathcal{B}_n^c} \leq \sum_{j\leq n}\mathbb{P}\left\{|E_{x^*}(j)\Delta E_w(j)|>1\right\} = o(1).
\]

\item The event $\mathcal{C}_n$: The edge $(\sigma^*(j+1), \sigma^*(j+k+2))$ is not in $x^*$. Therefore $\mathbb{P}\{w(\sigma^*(j+1), \sigma^*(j+k+2))=1\}=q=o(n^{-3/2})$. Thus $\prob{\mathcal{C}_n^c}=o(1)$ follows from the union bound.
\end{itemize}
\end{proof}


\section*{Acknowledgment}
J.~Ding is supported in part by the NSF Grant DMS-1757479.
Y.~Wu is supported in part by the NSF Grant CCF-1900507, an NSF CAREER award CCF-1651588, and an Alfred Sloan fellowship.
J.~Xu is supported by the NSF Grants IIS-1838124, CCF-1850743, and CCF-1856424. 
D.~Yang is supported by the NSF Grant CCF-1850743.
We thank David Pollard for his comments and suggestions.
 
\bibliographystyle{alpha}
\bibliography{kNN}

\end{document}